\documentclass[runningheads]{llncs}
\usepackage[english]{babel}
\usepackage[T1]{fontenc}
\usepackage{amsfonts, amsmath, amssymb,latexsym}
\usepackage[letterpaper,top=2cm,bottom=2cm,left=3cm,right=3cm,marginparwidth=1.75cm]{geometry}
\usepackage{graphicx, float}
\usepackage[colorlinks=true, allcolors=blue]{hyperref}

\newcommand{\Tr}{\operatorname{Tr}}
\newcommand{\erf}{\operatorname{erf}}

\begin{document}

\title{A Generalized Approach to Root-based Attacks against PLWE}
\author{Iv\'an Blanco Chac\'on\inst{1}\orcidID{0000-0002-4666-019X} \and
Ra\'ul Dur\'an D\'{\i}az\inst{2}\orcidID{0000-0001-6217-4768} \and \\
Rodrigo Mart\'{\i}n S\'anchez-Ledesma\inst{3,4}\orcidID{0009-0001-1845-2959}}

\institute{Departamento de F\'{\i}sica y Matem\'aticas, Universidad de Alcal\'a, Spain \\
\email{ivan.blancoc@uah.es} \and
Departamento de Autom\'atica, Universidad de Alcal\'a, Spain \\ \email{raul.duran@uah.es}
\and Departamento de \'Algebra, Universidad Complutense de Madrid, Spain \\
\email{rodrma01@ucm.es} \and
Indra Sistemas de Comunicaciones Seguras, Spain \\
\email{rmsanchezledesma@indra.es}}

\maketitle
\def\figurename{Algorithm}

\begin{abstract}
In the present work we address the robustness of the Polynomial Learning With
Errors problem extending previous results in \cite{BDNB:2023:TBC} and in
\cite{ELOS:2015:PWI}. In particular, we produce two kinds of new distinguishing
attacks: a) we generalize \cite{BDNB:2023:TBC} to the case where the defining
polynomial has a root of degree up to $4$, and b) we widen and refine the most
general attack in \cite{ELOS:2015:PWI} to the non-split case and determine
further dangerous instances previously not detected. Finally, we exploit our
results in order to show vulnerabilities of some cryptographically relevant
polynomials.

\end{abstract}
\keywords{PLWE \and Number Theory \and Algebraic Roots \and Trace-based Cryptanalysis.}

\section{Introduction}
The Ring Learning With Errors problem (RLWE) and the Polynomial Learning With
Errors problem (PLWE) sustain a large number of lattice-based cryptosystems
highly believed to be quantum resistant. Some of the strongest theoretical clues
pointing in this direction are the \emph{worst case-average case} reductions from an
approximate version of the Shortest Vector Problem on ideal lattices to the RLWE
problem, established in \cite{LPR:2013:ILL}, and to the PLWE problem over
power-of-two cyclotomic polynomials, established in \cite{SSTX:2009:EPK}.

Furthermore, for a large family of monic irreducible polynomials
$f(x)\in\mathbb{Z}[x]$, it is now well known that the PLWE problem for the ring
$R_q:=\mathbb{F}_q[x]/(f(x))$ is equivalent to the RLWE problem for the ring
$\mathcal{O}_f/q\mathcal{O}_f$ where $\mathcal{O}_f$ is the ring of integers of
the splitting field of $f(x)$ and $q$ is a suitable rational prime. By
\emph{equivalent} it is understood that a solution to the first problem can
be turned into a solution of the second (and vice versa) by an algorithm of
polynomial complexity in the degree of $f(x)$, causing a noise increase which
is also polynomial in the degree (see \cite[Section 4]{RSW:2018:RPP}).

Cryptographically relevant families for which both problems admit simultaneously
a \emph{worst case\-average case} reduction include the prime power cyclotomic case
(\cite{SSS:2021:CNV}, \cite{DD:2012:RPR}) or, more in general, cyclotomic fields whose
conductor is divisible by a small number of primes (\cite{Blanco:2020:REC}), as well as
the maximal totally real cyclotomic subextensions of conductor $2^rpq$ for
$r\geq 2$ and $p$, $q$ primes or $1$ (\cite{BL:2022:RPE}).

However, despite the aforementioned worst case reductions, several sets of
instances have been proved to be insecure in a number of works. For instance,
in \cite{ELOS:2015:PWI} and \cite{ELOS:2016:RCN}, a set of conditions has been
identified granting an efficient attack against the search version of PLWE
(hence against RLWE whenever they are equivalent). These conditions are of
three types: \emph{a}) the existence of an $\mathbb{F}_q$-root of $f(x)$ of
\emph{small multiplicative order}; \emph{b}) the existence of an $\mathbb{F}_q$-root
of \emph{small canonical residue} in $\{0,\dotsc,q-1\}$; and \emph{c}) the existence of an
$\mathbb{F}_q$-root such that the evaluation modulo $q$ of the error
polynomials over such root falls in the interval $\left[-\frac{q}{4},\frac{q}
{4}\right)$ with probability beyond $1/2$.

Further, in \cite{CLS:2017:ASR} the authors introduce the Chi-Square
decisional attack on PLWE. This test detects non-uniformity of the samples
modulo some of the prime ideals which divide $q$. For Galois number fields,
this is enough to detect non-uniformity modulo $q$ by virtue of the Chinese
Remainder Theorem. Moreover, the authors show that their attack succeeds over
some cyclotomic rings and subrings (although the ring of integers of the
maximal totally real subextension is not listed therein). The attack requires
typically beyond $10000$ samples and may take a computation time up to
thousands of hours.

Still, in \cite{CIV:2016:PWIR}, the attacks of \cite{CLS:2017:ASR}, \cite{ELOS:2015:PWI}
and \cite{ELOS:2016:RCN} are dramatically improved for the family of trinomials
$x^n+ax+b$; indeed, the number of samples is reduced to $7$ and the success rate
is increased to $100\%$. Moreover, the authors justify that the reason why
the attack succeeds is that the error distribution of this family is very skewed
in certain directions, so that a moderate-to-large distortion between
the coordinate and canonical embedding grants the success of the attack for
small-to-moderate error sizes.

Finally, \cite{Peikert:2016:HIR} extends the above geometric justification to
all the vulnerable families in \cite{ELOS:2015:PWI} and shows why the occurrence
of this phenomenon suffices for the attack to work. In short, it shows that a
generalized statistical attack works if there exists an element of small enough
norm in the dual of a prime ideal dividing $q$
(see \cite[Lemma 3]{Peikert:2016:HIR}).

These attacks, however, do not apply to any of the three lattice-based NIST
selections for standardization (ML-KEM \cite{FIPS:2024:203},
ML-DSA \cite{FIPS:2024:204} and FN-DSA) as they rely either on particular
instances of Module Learning With Errors resistant problem, or on NTRU lattices,
both not the target of these attacks. Moreover, these attacks work for the
non-dual version of RLWE and for small values of the Gaussian parameter of the
error distribution, while the worst case reduction of \cite{LPR:2013:ILL} holds
for the dual version, where a very precise lower bound for the Gaussian
parameter is required.

A more worrying  future potential threat to real-world cryptosystems, even if
theoretical, is the line of research which addresses the hardness of
the ideal lattice Shortest Vector Problem itself. In this spirit, the works
\cite{CDW:2017:SSC} and \cite{DPW:2019:SVF} exploit the Stickelberger class relations on
cyclotomic fields to reduce the complexity of this problem to
$\tilde{\mathcal{O}}(n^{1/2})$ and provide an explicit computation of the
asymptotic constant involved.

The present work started as an extension of \cite{BBDN:2023:CPB} and
\cite{BDNB:2023:TBC}, which exploit the existence of a root $\alpha$ of $f(x)$
over finite non-trivial $\mathbb{F}_q$-extensions with small order trace. These
works are in turn generalizations of \cite{CLS:2017:ASR} and \cite{ELOS:2016:RCN},
which assume the existence of small order roots in $\mathbb{F}_q$.

In both \cite{BBDN:2023:CPB} and \cite{BDNB:2023:TBC}, the key tool is to use
the composition of the evaluation-at-$\alpha$ with the trace map, in order to
project $R_q$ to a subring $R_{q,0}$ of large enough dimension over which the
PLWE problem can be solved in expected polynomial time with overwhelming
probability of success. In \cite{BBDN:2023:CPB}, the root is assumed to be
quadratic and the attack on $R_{q,0}$-samples is pulled back to $R_q$, allowing
for an attack on the usual PLWE, also with overwhelming probability of success
in expected polynomial time. In \cite{BDNB:2023:TBC}, an attack on
$R_{q,0}$-samples is obtained for cyclotomic polynomials of prime-power
conductor but it remains still unclear how this attack can be pulled back to
the original ring.

Hence, our original motivation was to extend these attacks to the case where
there exists a root of higher order and where the attack can be pulled back to
the whole ring. After succeeding in this task, we soon also observed that our
trace-based approach allows to extend all the attacks of \cite{CLS:2017:ASR} and
\cite{ELOS:2016:RCN} (not just those based on the root orders) and, moreover,
that it is still possible to refine one of the original algorithms and
extend it to the higher-order root case in order to identify additional
vulnerable instances of the PLWE problem.

Our work is organized as follows:

After some preliminaries in Section~\ref{s.2},
Section~\ref{s.3} reintroduces the framework of \cite{ELOS:2015:PWI},
considering both truncation and no truncation on the range of the Gaussian
distributions to appear in the PLWE problems. We believe this additional
consideration helps understand the true nature of the attacks and might prove
to be helpful when applied to real PLWE instances.

This section also introduces a revision of certain probability values assumed on
\cite{ELOS:2015:PWI,ELOS:2016:RCN} thus providing not only a better understanding
of the applicability of the attacks, but also revealing further instances that
yield to such attacks.

In Section~\ref{s.4} we extend the trace-based attack given in \cite{BBDN:2023:CPB}
for the quadratic root case to the higher-order root case
(Proposition~\ref{cvgen}), dealing with the case when the root is of small
order, thus giving rise to the \emph{small set attack}.

In Proposition~\ref{p.37} we prove that the expected time of our
algorithm is polynomial on the order of the root. We have tested our algorithm
with some explicit instances obtaining, for one of them, a success probability of
$0.629$ with $350$ test samples, and up to $0.993$ with $500$ samples
(see Subsection~\ref{sss.421}, instance $(2)$).

Finally, in Section~\ref{s.5} we apply the ideas of Section~\ref{s.4} to the
case when the evaluation of the error polynomials at the root gives values within
the interval $\left[-\frac{q}{4},\frac{q} {4}\right)$, which we term the
\emph{small values attack}.

\section{Preliminaries}\label{s.2}
Let $q$ be a prime and $f(x)\in \mathbb{Z}[x]$ a monic irreducible polynomial
over $\mathbb{Z}[x]$ of degree $N$. Denote by $R_q$ the ring
$\mathbb{F}_q[x]/(f(x))$. For the PLWE distribution, we assume a Gaussian
distribution of mean $0$ and a certain variance $\sigma^2$.

We introduce here the following auxiliary lemmas, that will be useful throughout this work:
\begin{lemma}\label{SumProb}
    Given a discrete uniform random variable $U$ in $\mathbb{F}_q$, where $q > 2$ is prime, and a
    discrete Gaussian random variable $E$ of mean 0 and variance
    $\sigma^2$ in $\mathbb{F}_q$, we have that $\mathrm{P}(U + E \mod q \in
[-\frac{q}{4}, \frac{q}{4})) = \frac{1}{2} \pm \frac{1}{2q}$.
\end{lemma}
\begin{proof}
Expanding over the possible values of $e\mod q$, we have
    \begin{align*}
        & \mathrm{P}(U + E \, \mod \, q \in [-\frac{q}{4}, \frac{q}{4})) = \\
        & \sum_{k = -\frac{q}{2}}^{\frac{q}{2}} \mathrm{P}(U + E \,\mod\, q \in
[-\frac{q}{4}, \frac{q}{4}) \,|\, E \,\mod\, q = k)\cdot \mathrm{P}(E
\,\mod\, q = k) = \\
        & \sum_{k = -\frac{q}{2}}^{\frac{q}{2}} \mathrm{P}(U \,\mod\,q \in
[-\frac{q}{4} - k, \frac{q}{4} -k)) \cdot \mathrm{P}(E \,\mod\, q = k)
    \end{align*}
    Now, it is required to analyze, given $k \in [-\frac{q}{2},
\frac{q}{2})$, $\mathrm{P}(U \,\mod\,q \in
[-\frac{q}{4} - k, \frac{q}{4} -k))$. This analysis must be done in terms of the value of $q$, specifically regarding its congruence modulo $4$:

If $q \equiv 1$ mod $4$, then $\exists p$ such that $q = 4p+1$:
    \begin{align*}
        & \mathrm{P}(U \,\mod\,q \in [-\frac{q}{4} - k, \frac{q}{4} -k)) =
\mathrm{P}(U \,\mod\,q \in [-p-k -\frac{1}{4}, p-k +\frac{1}{4}))
    \end{align*}
which means that $U$ can take any value in the (discrete) interval $[-p-k,
p-k]$, which represents $2p+1$ values in $\mathbb{F}_q$. That means that
the probability is equal to $$\frac{2p+1}{q} = \frac{\frac{q-1}{2} + 1}{q} = \frac{1}{2} + \frac{1}{2q}$$

On the other hand, if $q \equiv 3$ mod $4$, then $\exists p$ such that $q
= 4p+3$ and the same argument yields that the probability is equal to
$$\frac{2p+1}{q} = \frac{\frac{q-3}{2} + 1}{q} = \frac{1}{2} - \frac{1}
{2q}$$

Therefore, we have
    \begin{align*}
        &\sum_{k = \frac{q}{2}}^{\frac{q}{2}} \mathrm{P}(U \,\mod\,q \in
[-\frac{q}{4} - k, \frac{q}{4} -k)) \cdot \mathrm{P}(E \,\mod\, q = k) = \\
        & (\frac{1}{2} \pm \frac{1}{2q}) \cdot \sum_{k = -\frac{q}{2}}^{\frac{q}{2}}
\mathrm{P}(E
\,\mod\, q = k) = \frac{1}{2} \pm \frac{1}{2q}
    \end{align*}
    as the sum over the possible values of $E$, taken modulo $q$, will be in
the discrete interval $[-\frac{q}{2}, \frac{q}{2})$ and therefore that sum is
equal to 1.
\qed
\end{proof}
The proof of the above lemma yields the following corollary, which will also be employed throughout the work:
\begin{corollary}\label{CorSumProb}
    Given a discrete uniform random variable $U$ in $\mathbb{F}_q$, where $q > 2$ is prime, we have that
$\mathrm{P}(U \mod q \in [-\frac{q}{4}, \frac{q}{4}))
    = \frac{1}{2} \pm \frac{1}{2q}$.
\end{corollary}

\begin{lemma}\label{SumUnifGen}
If $X_1, X_2, \dotsc, X_n$ are independent uniform distributions over
$\mathbb{F}_{q^k}$ then, for each $\lambda_1$, $\lambda_2$, \ldots, $\lambda_n \in
\mathbb{F}_{q^k}$
not all of them zero, the variable $\sum_{i=1}^n \lambda_i \cdot X_i$ is
uniform
\end{lemma}
\begin{proof}
    The fact that the random variable $\lambda \cdot X$ is uniform in
    $\mathbb{F}_{q^k}$, if $\lambda \neq 0$ and $X$ is uniform in
    $\mathbb{F}_{q^k}$ is a direct consequence of working in a field.

    For $X_1, X_2$ independent uniforms in $\mathbb{F}_{q^k}$, we have
    \begin{align*}
        P[X_1 + X_2 = c] = \sum_{x_2 \in \mathbb{F}_{q^k}} P[X_1 = c - x_2 \wedge X_2 = x_2] = \frac{1}{q^k} \sum_{x_2 \in \mathbb{F}_{q^k}} P[X_2 = x_2] = \frac{1}{q^k}
    \end{align*}
    When put together, we arrive to the desired result.
\qed
\end{proof}

\section{Original \texorpdfstring{$\mathbb{F}_q$}{Lg}-based root attacks,
Revisited}\label{s.3}
In this section, we lay out the general setting that we will employ throughout
our work, in the hope that it will be both more intuitively and more closely
related to real-life deployments of PLWE-based schemes. This setting assumes
that $q$ is a prime and $f(x)\in \mathbb{Z}[x]$ a monic irreducible
polynomial over $\mathbb{Z}[x]$ of degree $N$. Denote by $R_q$ the ring
$\mathbb{F}_q[x]/(f(x))$.

We will treat separately the cases where the PLWE Gaussian distribution of mean
$0$ and variance $\sigma^2$ is either assumed to be truncated to width $2\sigma$
or not truncated, contrary to \cite{ELOS:2015:PWI,ELOS:2016:RCN}, which only consider
the truncated case. The rationale for this modification is that in real-life
implementations of PLWE-based schemes the distributions are often not
truncated. Moreover, we will see that this modification will impact the success
probability of the attacks.

To this end, let $\mathcal{G}_{\sigma}$ be a centered Gaussian random variable
$\mathcal{G}_{\sigma}$ of standard deviation $\sigma$ truncated at $2\sigma$ or
not. Let $p_0$ denote, from now on, the probability that $\mathcal{G}_{\sigma}$
takes values on the interval $[-2\sigma,2\sigma]$. If $\mathcal{G}_{\sigma}$ is
truncated to $2\sigma$ then $p_0=1$ and otherwise $p_0\approx 0.954500$ (with
error less than $10^{-6}$).

\subsection{Small Set Attack, Revisited}\label{ss.31}
The authors of \cite{ELOS:2015:PWI} give an attack on the decisional version of
PLWE for $R_q$
whenever a root $\alpha\in\mathbb{F}_q$ of $f(x)$ with (small) order $r$ in
$\mathbb{F}_q^\ast$ exists such that the cardinal of the set of all error
polynomials evaluated at $\alpha$, which we denote by $\Sigma$, is small
enough for the set to be swept through.

Algorithm~\ref{SSAAlg_Fq}, as proposed in \cite{ELOS:2015:PWI}, is given a
collection $S$ of $M$ samples, which come either from a uniform
distribution, $U$, or from a PLWE distribution, $\mathcal{G}_\sigma$,
assuming that $\mathrm{P}(D=\mathcal{G}_\sigma) = \mathrm{P}(D=U) = 1/2$,
i.e., both distributions are equiprobable. The Algorithm decides which of
both distributions the samples in collection $S$ were drawn from, either from
PLWE (returning a guess, $g$, for the evaluation of the secret $s(\alpha)$),
or from the uniform one. Sometimes the algorithm is not able to decide, then
returning \textbf{NOT ENOUGH SAMPLES}.

\begin{figure}[ht]
\centering
\begin{tabular}[c]{ll}
\hline
\textbf{Input:} & A collection of samples $S=\{(a_i(x),b_i(x))\}
_{i=1}^M\subseteq R_q^2$ \\
 & A look-up table $\Sigma$ of all possible values for $e(\alpha)$ \\
\textbf{Output:} & A guess $g\in\mathbb{F}_q$ for $s(\alpha)$,\\
 & or \textbf{NOT PLWE},\\
 & or \textbf{NOT ENOUGH SAMPLES}\\
\hline
\end{tabular}

\medskip

\begin{itemize}
\item $G:=\emptyset$
\item \texttt{\emph{for}} $g\in\mathbb{F}_q$ \texttt{\emph{do}}
	\begin{itemize}
	\item \texttt{\emph{for}} $(a_i(x),b_i(x))\in S$  \texttt{\emph{do}}
		\begin{itemize}
		\item \texttt{\emph{if}} $b_i(\alpha)-a_i(\alpha)g\notin \Sigma$
            \texttt{\emph{then next}} $g$
		\end{itemize}
	\item $G:=G\cup\{g\}$
	\end{itemize}
\item \texttt{\emph{if}} $G=\emptyset$ \texttt{\emph{then return}} \textbf{NOT
PLWE}
\item \texttt{\emph{if}} $G=\{g\}$ \texttt{\emph{then return}} $g$
\item \texttt{\emph{if}} $|G|>1$ \texttt{\emph{then return}} \textbf{NOT ENOUGH
SAMPLES}
\end{itemize}
\hrule
\caption{Algorithm for Small Set Attack over $\mathbb{F}_q$}
\label{SSAAlg_Fq}
\end{figure}

The reason for the algorithm's applicability is that if the order $r$ of the
root $\alpha$ is low and the discrete Gaussian distribution is truncated, then the
cardinal of the set $\Sigma$ of plausible values for $e(\alpha)$ as $e(x)$ runs
over the error distribution is small enough with respect to $|\mathbb{F}_q|$.
Remark that
\begin{equation*}
 e(\alpha)= \sum_{j=0}^{r-1}e^\ast_j\alpha^j\mbox { where }e^\ast_j=\sum_{i=0}^
{\left\lfloor\frac{N}{r}\right\rfloor-1}e_{ir+j},
\end{equation*}
where $\lfloor\phantom{x}\rfloor$ stands for the floor function.

Notice that each term $e^\ast_j$ follows a Gaussian distribution of mean $0$
and variance $\overline{\sigma}^2 = \left\lfloor\frac{N}{r}\right\rfloor \cdot
\sigma^2$.  Therefore, we can ensure that for each $j\in\{0,...,r-1\}$, we have
that $|e^\ast_j|\leq 2\overline{\sigma}$ with probability at least $p_0$.

As explained above, let us define
$$
\Sigma:=\left\{\sum_{j=0}^{r-1}x_j\alpha^j:\,\,x_j\in[-2\overline{\sigma},2
\overline{\sigma}]\cap\mathbb{Z}\right\},
$$
namely, the set of all possibilities for $e(\alpha)$, such that the absolute
value of each block $e^\ast_j$ is bounded by $2\overline{\sigma}$.

To begin with, we will consider the case where Gaussian distributions are
truncated to width $2\sigma$, where $\sigma$ is their respective standard
deviation. Then, we have the following

\begin{proposition}\label{SSA_PostProb_Trunc_Fq}
Let $\Sigma$, $M$ and $r$ be defined as above, and assume $|\Sigma|<q$. Then
\begin{enumerate}
\item
If Algorithm~\ref{SSAAlg_Fq} returns \textbf{PLWE} or \textbf{NOT ENOUGH SAMPLES},
then the samples are \textbf{PLWE} with probability at least $1-q\left(\frac
{|\Sigma|}{q}\right)^M$.
\item
If Algorithm~\ref{SSAAlg_Fq} returns \textbf{NOT PLWE}, then the samples are
uniform with probability $1$.
\end{enumerate}
\end{proposition}
\begin{proof}
Let us prove both items in turn.
\begin{enumerate}
\item
Denote with $E$ the event ``Algorithm~\ref{SSAAlg_Fq} returns \textbf{PLWE} or
\textbf{NOT ENOUGH SAMPLES}''. Then, applying Bayes' theorem,
\begin{eqnarray}\label{eq.1}
\mathrm{P}(D=\mathcal{G}_\sigma | E) = 1 - \mathrm{P}(D=U|E) & = & 1 -
\frac{\mathrm{P}(E|D=U)\cdot\mathrm{P}(D=U)}
{\mathrm{P}(E|D=U)\cdot\mathrm{P}(D=U)+
\mathrm{P}(E|D=\mathcal{G}_\sigma)\cdot\mathrm{P}(D=\mathcal{G}_\sigma)}
\nonumber \\
& = & 1 - \frac{\mathrm{P}(E|D=U)}{\mathrm{P}(E|D=U) +
\mathrm{P}(E|D=\mathcal{G}_\sigma)}.
\end{eqnarray}
Now remark that $\mathrm{P}(E|D=\mathcal{G}_\sigma)$ is the probability that
at least a $g\in\mathbb{F}_q$ exists such that $b_i(\alpha)-a_i(\alpha)g\in
\Sigma$, for all $i=1,\dotsc,M$, i.e., for all samples in $S$. Since we are
supposing that the distribution is $\mathcal{G}_\sigma$ and Gaussian
distributions to be truncated by $2\sigma$, at least one $g$ is sure to exist,
namely, $s(\alpha)$, the evaluation of the secret. Then we have that
$\mathrm{P}(E|D=\mathcal{G}_\sigma)=1$.

As for $\mathrm{P}(E|D=U)$, let $U_g$ be the event that $b_i(\alpha)-a_i
(\alpha)g\in\Sigma$, for all $i=1,\dotsc,M$, i.e., for all samples in $S$. Then
\[
\mathrm{P}(E|D=U) = \mathrm{P}\left(\bigcup_{g=0}^{q-1} U_g\right) \leq
\sum_{g=0}^{q-1}\mathrm{P}(U_g),
\]
where the inequality comes from the fact that we cannot guarantee that the
events $U_g$ be independent for different values of $g$.

Since we are now in the case where $a_i$ and $b_i$ are uniformly randomly
selected, so $b_i-a_ig$ is also uniformly random according to
Lemma~\ref{SumUnifGen}; this means that $\mathrm{P}(b_i-a_ig\in\Sigma) =
|\Sigma|/q$ for any $i$ and any $g$, so
\[
\mathrm{P}(U_g) = \prod_{i=1}^M \frac{|\Sigma|}{q} = \left(\frac{|\Sigma|}{q}
\right)^M,
\]
and thus we have
\[
\mathrm{P}(E|D=U) \leq q\left(\frac{|\Sigma|}{q}\right)^M.
\]
Plugging all the results above into Equation~\eqref{eq.1}, we eventually obtain
\[
\mathrm{P}(D=\mathcal{G}_\sigma | E) = 1 - \frac{\mathrm{P}(E|D=U)}
{\mathrm{P}(E|D=U)+1} \geq 1 - \mathrm{P}(E|D=U) \geq 1 -
q\left(\frac{|\Sigma|}{q} \right)^M.
\]
\item
Denote now with $E^\prime$ the event ``Algorithm~\ref{SSAAlg_Fq} returns \textbf{NOT
PLWE}''. Applying again Bayes' theorem,
\begin{eqnarray*}
\mathrm{P}(D=U|E^\prime) & = & \frac{\mathrm{P}(E^\prime|D=U)\cdot\mathrm{P}(D=U)}
{\mathrm{P}(E^\prime|D=U)\cdot\mathrm{P}(D=U)+
\mathrm{P}(E^\prime|D=\mathcal{G}_\sigma)\cdot\mathrm{P}(D=\mathcal{G}_\sigma)}
\nonumber \\
& = & \frac{\mathrm{P}(E^\prime|D=U)}{\mathrm{P}(E^\prime|D=U) +
\mathrm{P}(E^\prime|D=\mathcal{G}_\sigma)}.
\end{eqnarray*}
Now in order to compute $\mathrm{P}(E^\prime|D=\mathcal{G}_\sigma)$ remark
that the event $E^\prime$ is complementary to the event $E$, defined in the
previous item. Hence
\[
\mathrm{P}(E^\prime|D=\mathcal{G}_\sigma) = 1 - \mathrm{P}(E|D=\mathcal{G}
_\sigma) = 1 - 1 = 0.
\]
Consequently,
\[
\mathrm{P}(D=U|E^\prime) = 1,
\]
as required.
\end{enumerate}
\qed
\end{proof}

We deal now with the case where the Gaussian distribution is not truncated.
Then we have the following
\begin{proposition}\label{SSA_PostProb_NotTrunc_Fq}
Let $\Sigma$, $M$ and $r$ be defined as above, and assume $|\Sigma|<qp_0^r$. Then
\begin{enumerate}
\item
If Algorithm~\ref{SSAAlg_Fq} returns \textbf{PLWE} or \textbf{NOT ENOUGH SAMPLES},
then the samples are \textbf{PLWE} with probability at least $1-q\left(\frac
{|\Sigma|}{qp_0^r}\right)^M$.
\item
If Algorithm~\ref{SSAAlg_Fq} returns \textbf{NOT PLWE}, then the samples are
uniform with probability at least $1/2$ for large enough $M$.
\end{enumerate}
\end{proposition}
\begin{proof}
Let us prove both items in turn.
\begin{enumerate}
\item
As before, denote with $E$ the event ``Algorithm~\ref{SSAAlg_Fq} returns
\textbf{PLWE} or \textbf{NOT ENOUGH SAMPLES}''. Applying again Bayes' theorem,
\[
\mathrm{P}(D=\mathcal{G}_\sigma | E) = 1 - \mathrm{P}(D=U|E) =
1 - \frac{\mathrm{P}(E|D=U)}{\mathrm{P}(E|D=U) +
\mathrm{P}(E|D=\mathcal{G}_\sigma)}
\geq 1 - \frac{\mathrm{P}(E|D=U)}{\mathrm{P}(E|D=\mathcal{G}_\sigma)}.
\]
Let $U_g$ be again the event that $b_i(\alpha)-a_i (\alpha)g\in\Sigma$, for
all $i=1,\dotsc,M$, i.e., for all samples in $S$. Then
\[
\mathrm{P}(E|D=\mathcal{G}_\sigma) = \mathrm{P}\left(\bigcup_{g=0}^{q-1}
U_g\right) \geq \mathrm{P}(U_{g^\ast}),
\]
where $g^\ast = s(\alpha)$. Observe that $b_i(\alpha)-a_i (\alpha)g^\ast =
e_i(\alpha)$. Since we are now dealing with the not truncated case, we have
\[
\mathrm{P}(U_{g^\ast}) = \prod_{i=1}^M \mathrm{P}(e_i(\alpha)\in\Sigma) \geq p_0^{Mr}.
\]
Moreover, $\mathrm{P}(E|D=U)$ is the same as in the truncated case, since the
uniform distribution is not affected by the truncation. Combining results,
\[
\mathrm{P}(D=\mathcal{G}_\sigma | E) \geq 1 - \frac{\mathrm{P}(E|D=U)}
{\mathrm{P}(E|D=\mathcal{G}_\sigma)} \geq 1 - \frac{\mathrm{P}(E|D=U)}{p_0^{Mr}}
\geq 1 - \frac{q\left(\frac{|\Sigma|}{q}\right)^M}{p_0^{Mr}} = 1 -
q\left(\frac{|\Sigma|}{qp_0^r}\right)^M.
\]
\item
Denote again with $E^\prime$ the event ``Algorithm~\ref{SSAAlg_Fq} returns \textbf{NOT
PLWE}''. Applying Bayes' theorem,
\[
\mathrm{P}(D=U|E^\prime) = \frac{\mathrm{P}(E^\prime|D=U)}{\mathrm{P}(E^\prime|D=U) +
\mathrm{P}(E^\prime|D=\mathcal{G}_\sigma)}.
\]
Remark as before that the event $E^\prime$ is complementary to the event $E$,
defined in the previous item. Hence
\[
\mathrm{P}(E^\prime|D=\mathcal{G}_\sigma) = 1 - \mathrm{P}(E|D=\mathcal{G}
_\sigma) \leq 1 - p_0^{Mr}.
\]
Consequently,
\[
\mathrm{P}(D=U|E^\prime) \geq \frac{\mathrm{P}(E^\prime|D=U)}{\mathrm{P}
(E^\prime|D=U) + 1-p_0^{Mr}},
\]
which is greater than $1/2$ if and only if $\mathrm{P}
(E^\prime|D=U)\geq 1-p_0^{rM}$. Since
\[
\mathrm{P}(E^\prime|D=U) \geq 1-q\left(\frac{|\Sigma|}{q}\right)^M,
\]
for the conclusion to hold it is enough to take $M$ such that $q^{\frac{1}{M}}
\frac{|\Sigma|}{q}\leq p_0^r$.
\end{enumerate}
\qed
\end{proof}

\subsubsection{Success probabilities.}
Observe that from the discussion above, we can derive the success probability
of Algorithm~\ref{SSAAlg_Fq}. For the truncated case, we have the following
\begin{proposition}\label{SSA_Success_Trunc_Fq}
With the same notations and assumptions as in Proposition~\ref{SSA_PostProb_Trunc_Fq},
\begin{enumerate}
\item
If the samples are PLWE, Algorithm~\ref{SSAAlg_Fq} guesses correctly with
probability $1$.
\item
If the samples are uniform, Algorithm~\ref{SSAAlg_Fq} guesses correctly with
probability at least $1-q\left(\frac{|\Sigma|}{q}\right)^M$.
\end{enumerate}
\end{proposition}
\begin{proof}
The first item is equivalent to computing $\mathrm{P}(E|D=\mathcal{G}
_\sigma)$, which is $1$, according to the proof in
Proposition~\ref{SSA_PostProb_Trunc_Fq}. Analogously, the second item corresponds to
computing $\mathrm{P}(E^\prime|D=U)$, which in the same place is shown to be
greater or equal to $1-q\left(\frac{|\Sigma|}{q}\right)^M$.
\qed
\end{proof}

Now for the not truncated case, we have the following
\begin{proposition}\label{SSA_Success_NotTrunc_Fq}
With the same notations and assumptions as in Proposition~\ref{SSA_PostProb_NotTrunc_Fq},
\begin{enumerate}
\item
If the samples are PLWE, Algorithm~\ref{SSAAlg_Fq} guesses correctly with
probability at least $p_0^{Mr}$.
\item
If the samples are uniform, Algorithm~\ref{SSAAlg_Fq} guesses correctly with
probability at least $1-q\left(\frac{|\Sigma|}{q}\right)^M$.
\end{enumerate}
\end{proposition}
\begin{proof}
The first item is equivalent to computing $\mathrm{P}(E|D=\mathcal{G}
_\sigma)$, which now is at least $p_0^{Mr}$, according to the proof in
Proposition~\ref{SSA_PostProb_NotTrunc_Fq}. The second item is not affected by the
truncation so the bound is the same as for the truncated case, namely,
greater or equal to $1-q\left(\frac{|\Sigma|}{q}\right)^M$.
\qed
\end{proof}

\subsubsection{An improvement for a higher number of samples.}\label{HNS_SSA}

Remark that in the attack described above, one of the success probabilities
actually decreases with the number of samples. Although shocking at first, this
is a natural consequence of the non-truncated behavior, as with each sample we
are asking for another tentative error value to maintain the truncated behavior,
which is less likely to happen.

While the attack itself could be negatively influenced by a higher
number of samples (in case the samples were originally PLWE
distributed), this does not mean that additional information (i.e. more
samples) cannot be employed to assist in the decision.

One could always find an
acceptable value of $M := M_0$ (such that both success probabilities
remain as high as possible) and run the attack multiple times for
different combinations of $M_0$ samples, and make a decision based upon
some established rule. This is precisely what the remainder of this
section will provide.

The idea is to take Algorithm~\ref{SSAAlg_Fq} as a sub-process of our final
Decision Algorithm. Given $M$ samples, an appropriate number of samples
$M_0 < M$ will be chosen to be used within the sub-process. Then, it
will be run $c := \left\lfloor M/M_0\right\rfloor$ times, collecting the output
results from each round. Informally, if Algorithm~\ref{SSAAlg_Fq} outputs
anything different from \textbf{NOT PLWE} a number of times greater than
expected, the Decision Algorithm concludes that the samples were indeed PLWE.
Otherwise, the Decision Algorithm outputs that the samples were drawn from the
uniform distribution.

Thus, our first task is to decide on an acceptable threshold for the number of
expected executions of Algorithm~\ref{SSAAlg_Fq} as a sub-process in order to
output \textbf{PLWE}. In Proposition~\ref{SSA_Success_NotTrunc_Fq}, the success
probability of guessing the PLWE distribution of the samples is
characterized by:
\[
\mathrm{P}_{\mathcal{G}_{\sigma}}(M_0) := \mathrm{P}(E| D = \mathcal{G}_{\sigma}) \geq p_0^{M_0 r}.
\]
Therefore, given $M$ samples, the sub-process must be executed $c$ times, thus
providing an expected number of outputs, $T$, different from \textbf{NOT PLWE}
given by
\[
T := c \cdot \mathrm{P}_{\mathcal{G}_{\sigma}}(M_0) \geq c \cdot p_0^{M_0 r}
\]
in case the samples were actually PLWE.

With these values, we define the \emph{Extended Small Set Attack} in
Algorithm~\ref{ExtendedSSAAlg_Fq}.

\begin{figure}[H]
\centering
\begin{tabular}[c]{ll}
\hline
\textbf{Input:} & A collection of samples $S=\{(a_i(x),b_i(x))\}
                                           _{i=1}^M \in R_q^2$ \\
& A choice $M_0$ for the number of samples of the sub-process \\
\textbf{Output:} & A guess for the distribution of the samples, either \\
                 & \textbf{PLWE} or \textbf{UNIFORM}\\
\hline
\end{tabular}

\medskip

\begin{itemize}
\item $T:=\lceil c \cdot p_0^{M_0r} \rceil$
\item $C:=0$
\item \texttt{\emph{for}} $j\in \{0, ..., c - 1\}$ \texttt{\emph{do}}
	\begin{itemize}
	\item $result := \mathrm{SmallSetAttack}\left(S_{j} := \{(a_i(x),b_i(x))\}
                         _{i=j\cdot M_0 + 1}^{(j+1)\cdot M_0}\right)$
		\begin{itemize}
		\item \texttt{\emph{if}} $result$ $\neq$ \textbf{NOT PLWE}
            \texttt{\emph{then}}
			\begin{itemize}
			\item $C := C + 1$
			\end{itemize}
		\end{itemize}
	\end{itemize}
\item \texttt{\emph{if}} $C < T$ \texttt{\emph{then return}} \textbf{UNIFORM}
\item \texttt{\emph{else return}} \textbf{PLWE}
\end{itemize}
\hrule
\caption{Algorithm for Extended Small Set Attack}
\label{ExtendedSSAAlg_Fq}
\end{figure}

Under this attack, we have the following success prediction probabilities:
\begin{itemize}
\item
Case $D = \mathcal{G}_{\sigma}$: It is required to analyze the probability
$\mathrm{P}(C \geq T | D = \mathcal{G}_{\sigma})$. By definition, this can be
seen as $1 - \mathrm{P}(C < T | D = \mathcal{G}_{\sigma})$, which in turn can
be interpreted in terms of a cumulative binomial distribution. Therefore, we
have that
\[
\mathrm{P}(C \geq T | D = \mathcal{G}_{\sigma}) := 1 - F(T-1, c,
\mathrm{P}_{\mathcal{G}_{\sigma}}(M_0)) \geq 1 - F(T-1, c, p_0^{M_0 r})
\]
since the cumulative binomial distribution enjoys decreasing monotonicity on the
\emph{success probability} parameter.

\item
Case $D = U$: Now $\mathrm{P}(C < T | D = U)$ is yet again given
in terms of the cumulative binomial distribution, i.e., the probability of a
binomial distribution of parameters $n = c$ and $p = \mathrm{P}_{U}(M_0)$, to
output any of the values $\{0,1,2,\dotsc,T-1\}$. Therefore, we have:
\[
\mathrm{P}(C < T | D = U) := F(T-1, c, \mathrm{P}_{U} (M_0)) \geq F(T-1, c,
1 - \left(|\Sigma|/q\right)^{M_0})
\]
since the probability of success in case the samples were uniform (i.e.
$\mathrm{P}_{U} (M_0)$), in terms of the proof of Proposition
\ref{SSA_Success_NotTrunc_Fq}, verifies
\[
\mathrm{P}_{U} (M_0) = 1 - \mathrm{P}\left(\bigcup_{g=0}^{q-1} U_g\right)
\geq 1 - \mathrm{P}\left(U_g\right) = 1 - \left(|\Sigma|/q\right)^M, \forall
g \in \mathbb{F}_q
\]
and therefore can be
upper-bounded by the probability that a particular $g_0 \in
\mathbb{F}_q$ verifies the condition for $M = M_0$ samples, which is $1 - \left(|\Sigma|/q\right)^{M_0}$.
\end{itemize}

With the above characterization of the success probabilities of each
distribution, we prove that Algorithm~\ref{ExtendedSSAAlg_Fq} is successful:
\begin{proposition}\label{Extended_SSA_Success_Fq}
Let $C, T, M_0$ be as specified in Algorithm~\ref{ExtendedSSAAlg_Fq}. Then, if
$1 - \left(|\Sigma|/q\right)^{M_0} < p_0^{M_0 r}$, Algorithm~\ref{ExtendedSSAAlg_Fq}
is successful.
\end{proposition}
\begin{proof}
We need to prove that the algorithm verifies that:
\begin{multline*}
\frac{1}{2} \cdot (\mathrm{P}(C < T | D = U) + \mathrm{P}(C
\geq T | D = \mathcal{G}_{\sigma})) \geq \\
\frac{1}{2} + \frac{1}{2}\cdot (F(T-1, c,
1 - \left(|\Sigma|/q\right)^{M_0}) -F(T-1, c,
p_0^{M_0 r})) > \frac{1}{2}.
\end{multline*}

Employing the decreasing monotonicity of the cumulative binomial distribution,
this is true as long as
\[
1 - \left(|\Sigma|/q\right)^{M_0} < p_0^{M_0 r}.
\]
\qed
\end{proof}

This result provides a characterization of how $M_0$ can be chosen in
order for the attack to be applicable, although further refinements
could be employed.

It is important to note that this extended algorithm does not make direct use
of features of $\mathbb{F}_q$, other than those related to the algorithm being
employed as sub-process. Therefore, this technique will also be applicable in
the generalization to finite field extensions of Section~\ref{s.4}.

\subsection{Small Values Attack, Revisited}
The authors of \cite{ELOS:2015:PWI} give yet another attack
on the decisional version of PLWE for $R_q$ whenever a root
$\alpha\in\mathbb{F}_q$ of $f(x)$ exists such that the evaluations
of the error polynomials at $\alpha$ give ``small'' values.

Algorithm~\ref{SVAAlg_Fq}, as proposed in \cite{ELOS:2015:PWI}, is given a
collection $S$ of $M$ samples, which come either from a uniform
distribution, $U$, or from a PLWE distribution, $\mathcal{G}_\sigma$,
assuming that $\mathrm{P}(D=\mathcal{G}_\sigma) = \mathrm{P}(D=U) = 1/2$,
i.e., both distributions are equiprobable. The Algorithm decides which one of
both distributions the samples in collection $S$ were drawn from, either from
PLWE (returning a guess, $g$, for the evaluation of the secret $s(\alpha)$),
or from the uniform one. Sometimes the algorithm is not able to decide, then
returning \textbf{NOT ENOUGH SAMPLES}.

\begin{figure}[ht]
\centering
\begin{tabular}[c]{ll}
\hline
\textbf{Input:} & A collection of samples $S=\{(a_i(x),b_i(x))\}
_{i=1}^M \in R_q^2$ \\
\textbf{Output:} & A guess $g\in\mathbb{F}_q$ for $s(\alpha)$,\\
 & or \textbf{NOT PLWE},\\
 & or \textbf{NOT ENOUGH SAMPLES}\\
\hline
\end{tabular}

\medskip

\begin{itemize}
\item \texttt{\emph{set}} $G:=\emptyset$
\item \texttt{\emph{for}} $g\in \mathbb{F}_q$ \texttt{\emph{do}}
	\begin{itemize}
	\item \texttt{\emph{for}} $(a_i(x),b_i(x))\in S$  \texttt{\emph{do}}
		\begin{itemize}
		\item \texttt{\emph{if}} $b_i(\alpha)-a_i(\alpha)g \notin [-\frac{q}{4}, \frac{q} {4})$ \texttt{\emph{then}} \texttt{\textbf{next}} $g$
		\end{itemize}
        \item \texttt{\emph{set}} $G:=G\cup\{g\}$
	\end{itemize}
    \item \texttt{\emph{if}} $G=\emptyset$ \texttt{\emph{then return}} \textbf{NOT
    PLWE}
    \item \texttt{\emph{if}} $G=\{g\}$ \texttt{\emph{then return}} $g$
    \item \texttt{\emph{if}} $|G|>1$ \texttt{\emph{then return}} \textbf{NOT ENOUGH
    SAMPLES}
\end{itemize}
\hrule
\caption{Algorithm for Small Values Attack over $\mathbb{F}_q$}
\label{SVAAlg_Fq}
\end{figure}
The reason of the algorithm's applicability is that the evaluation of the error
polynomials at $\alpha$ yields new Gaussian distributions of mean $0$ and a
certain variance $\overline{\sigma}^2$. Moreover, the smaller the order of the
root, the smaller the variance of the image Gaussian distribution.

\subsubsection{Case \texorpdfstring{$\alpha = \pm 1$.}{Lg}}\label{case1}
In this case, the value $e_i(\alpha) = \sum_{j=0}^{n-1} e_{ij}\cdot \alpha^j$
is a sum of $n$ Gaussians of mean $0$ and variance $\sigma^2$, so it is a
Gaussian of mean $0$ and variance $n \cdot \sigma^2$. Therefore
$\mathcal{G}_{\overline{\sigma}}$ is a Gaussian with $ \overline{\sigma}^2 = n
\cdot \sigma^2$.

\subsubsection{Case \texorpdfstring{$\alpha \neq \pm 1$ and has small order
$r$ modulo $q$.}{Lg}}\label{case2}
Here, the small multiplicative order of the root $\alpha$ is used to group the error
value $e(\alpha)$ into $r$ packs of $\frac{n}{r}$ values. Each pack
forms an independent Gaussian of mean $0$ and variance $\frac{n}{r} \cdot
\sigma^2$, and thus the overall distribution is a weighted sum of such Gaussians,
which generates a Gaussian of mean $0$ and variance $\overline{\sigma}^2 =
\sum_{i=0}^{r-1} \frac{n}{r} \cdot \sigma^2 \cdot \alpha^{2\cdot i}$.

\subsubsection{Case \texorpdfstring{$\alpha \neq \pm 1$ and does not have
small order modulo $q$.}{Lg}}\label{case3}
In the most general case, we just have a weighted sum of $n$ Gaussians,
which is itself a Gaussian of mean $0$ and variance $\overline{\sigma}^2 =
\sum_{i=0}^{n-1} \sigma^2 \cdot \alpha^{2\cdot i}$.

Under this attack the key factor is that, if the bound that reigns over the
Gaussian image does not surpass the value $\frac{q}{4}$, then the probability
that for any sample $i$ the evaluation of the error polynomial lies inside
$[-\frac{q}{4}, \frac{q}{4}$) is $1$ for $g = s(\alpha)\bmod{q}$, under the
supposition that the resulting Gaussian distribution is truncated.

Below we will deal in turn with both truncated and not truncated Gaussian
distributions. To begin with, we deal with the first one by stating the
following

\begin{proposition}\label{SVA_PostProb_Trunc_Fq}
Assume $2\overline{\sigma} \leq \frac{q}{4}$ and let $M$ be the number of input
samples. Then
\begin{enumerate}
\item
If Algorithm~\ref{SVAAlg_Fq} returns \textbf{PLWE} or \textbf{NOT ENOUGH SAMPLES},
then the samples are \textbf{PLWE} with probability at least
$1-q\left(\frac{1}{2}\pm\frac{1}{2q}\right)^M$.
\item
If Algorithm~\ref{SVAAlg_Fq} returns \textbf{NOT PLWE}, then the samples are
uniform with probability $1$.
\end{enumerate}
\end{proposition}
\begin{proof}
We follow essentially the same lines as in the proof of
Proposition~\ref{SSA_PostProb_Trunc_Fq}, dealing with both items in turn.
\begin{enumerate}
\item
Denote with $E$ the event ``Algorithm~\ref{SVAAlg_Fq} returns \textbf{PLWE} or
\textbf{NOT ENOUGH SAMPLES}''. Then, applying Bayes' theorem,
\begin{equation}\label{eq.2}
\mathrm{P}(D=\mathcal{G}_\sigma | E) = 1 - \mathrm{P}(D=U|E) =
1 - \frac{\mathrm{P}(E|D=U)}{\mathrm{P}(E|D=U) +
\mathrm{P}(E|D=\mathcal{G}_\sigma)}.
\end{equation}
Now remark that $\mathrm{P}(E|D=\mathcal{G}_\sigma)$ is the probability that
at least a $g\in\mathbb{F}_q$ exists such that $b_i(\alpha)-a_i(\alpha)g\in
[-\frac{q}{4},\frac{q}{4})$, for all $i=1,\dotsc,M$, i.e., for all samples in
$S$. Since we are supposing that the distribution is $\mathcal{G}_\sigma$, at
least one $g$ is sure to exist, namely, $g=s(\alpha)$, the evaluation of the
secret, as $e_i(\alpha)\in [-2\overline{\sigma},2\overline{\sigma}]\subseteq
[-\frac{q}{4},\frac{q}{4})$ for each sample $i$. Then we have that
$\mathrm{P}(E|D=\mathcal{G}_\sigma)=1$.

As for $\mathrm{P}(E|D=U)$, let $U_g$ be the event that $b_i(\alpha)-a_i
(\alpha)g\in [-\frac{q}{4},\frac{q}{4})$, for all $i=1,\dotsc,M$, i.e., for
all samples in $S$. Then
\[
\mathrm{P}(E|D=U) = \mathrm{P}\left(\bigcup_{g=0}^{q-1} U_g\right) \leq
\sum_{g=0}^{q-1}\mathrm{P}(U_g),
\]
where the inequality comes from the fact that we cannot guarantee that the
events $U_g$ be independent for different values of $g$.

Since we are now in the case where $a_i$ and $b_i$ are uniformly randomly
selected, so $b_i-a_ig$ is also uniformly random according to
Lemma~\ref{SumUnifGen}. By applying Proposition~\ref{SumProb}, this means
that for any $i$ and any $g$, $\mathrm{P}(b_i-a_ig\in [-\frac{q}{4},
\frac{q}{4})) = \frac{1}{2}\pm\frac{1}{2q}$, so
\[
\mathrm{P}(U_g) = \prod_{i=1}^M \frac{1}{2}\pm\frac{1}{2q} = \left(\frac{1}{2}
\pm\frac{1}{2q}\right)^M,
\]
and thus we have
\[
\mathrm{P}(E|D=U) \leq \sum_{g=0}^{q-1}\mathrm{P}(U_g) = q\left(\frac{1}{2}\pm
\frac{1}{2q}\right)^M.
\]
Plugging all the results above into Equation~\eqref{eq.2}, we eventually obtain
\[
\mathrm{P}(D=\mathcal{G}_\sigma | E) = 1 - \frac{\mathrm{P}(E|D=U)}
{\mathrm{P}(E|D=U)+1} \geq 1 - \mathrm{P}(E|D=U) \geq 1 -
q\left(\frac{1}{2}\pm\frac{1}{2q}\right)^M.
\]
\item
The proof for this item is exactly the same as the one for item~2 in
Proposition~\ref{SSA_PostProb_Trunc_Fq}.
\end{enumerate}
\qed
\end{proof}

Now we turn our attention to the not truncated case. We have the following
\begin{proposition}\label{SVA_PostProb_NotTrunc_Fq}
Assume $2\overline{\sigma} \leq \frac{q}{4}$ and $\frac{1}{2}\pm\frac{1}{2q}
<p_0$. Let $M$ be the number of input samples. Then
\begin{enumerate}
\item
If Algorithm~\ref{SVAAlg_Fq} returns \textbf{PLWE} or \textbf{NOT ENOUGH SAMPLES},
then the samples are \textbf{PLWE} with probability at least $1-q\left(\frac
{\frac{1}{2} \pm \frac{1}{2q}}{p_0}\right)^M$.
\item
If Algorithm~\ref{SVAAlg_Fq} returns \textbf{NOT PLWE}, then the samples are
uniform with probability at least $1/2$ for large enough $M$.
\end{enumerate}
\end{proposition}
\begin{proof}
We follow essentially the same lines as in the proof of
Proposition~\ref{SSA_PostProb_NotTrunc_Fq}, dealing with both items in turn.
\begin{enumerate}
\item
As before, denote with $E$ the event ``Algorithm~\ref{SVAAlg_Fq} returns
\textbf{PLWE} or \textbf{NOT ENOUGH SAMPLES}''. Applying again Bayes' theorem,
\[
\mathrm{P}(D=\mathcal{G}_\sigma | E) = 1 - \mathrm{P}(D=U|E) =
1 - \frac{\mathrm{P}(E|D=U)}{\mathrm{P}(E|D=U) +
\mathrm{P}(E|D=\mathcal{G}_\sigma)}
\geq 1 - \frac{\mathrm{P}(E|D=U)} {\mathrm{P}(E|D=\mathcal{G}_\sigma)}.
\]
Let $U_g$ be the event that $b_i(\alpha)-a_i (\alpha)g\in [-\frac{q}{4},
\frac{q}{4})$, for all $i=1,\dotsc,M$, i.e., for all samples in $S$. Then
\[
\mathrm{P}(E|D=\mathcal{G}_\sigma) = \mathrm{P}\left(\bigcup_{g=0}^{q-1}
U_g\right) \geq \mathrm{P}(U_{g^\ast}),
\]
where $g^\ast = s(\alpha)$. Observe that $b_i(\alpha)-a_i (\alpha)g^\ast =
e_i(\alpha)$. Since we are now dealing with the not truncated case, we have
\[
\mathrm{P}(U_{g^\ast}) = \prod_{i=1}^M \mathrm{P}(e_i(\alpha)\in
[-\frac{q}{4},\frac{q}{4})) \geq  p_0^M.
\]
Moreover, $\mathrm{P}(E|D=U)$ is the same as in the truncated case, since the
uniform distribution is not affected by the truncation. Combining results,
\[
\mathrm{P}(D=\mathcal{G}_\sigma | E) \geq 1 - \frac{\mathrm{P}(E|D=U)}
{\mathrm{P}(E|D=\mathcal{G}_\sigma)} \geq 1 - \frac{\mathrm{P}(E|D=U)}{p_0^M}
\geq 1 - \frac{q\left(\frac{1}{2}\pm\frac{1}{2q}\right)^M}{p_0^M}.
\]
\item
The proof for this item is exactly the same as the one for item~2 in
Proposition~\ref{SSA_PostProb_Trunc_Fq}.
\end{enumerate}
\qed
\end{proof}

\subsubsection{Success probabilities.}
Remark that from the discussion above, we can derive the success probability
of Algorithm~\ref{SVAAlg_Fq}. For the truncated case, we have the following
\begin{proposition}\label{SVA_Success_Trunc_Fq}
With the same notations and assumptions as in Proposition~\ref{SVA_PostProb_Trunc_Fq},
\begin{enumerate}
\item
If the samples are PLWE, Algorithm~\ref{SVAAlg_Fq} guesses correctly with
probability $1$.
\item
If the samples are uniform, Algorithm~\ref{SVAAlg_Fq} guesses correctly with
probability at least $1-q\left(\frac{1}{2}\pm\frac{1}{2q}\right)^M$.
\end{enumerate}
\end{proposition}
\begin{proof}
The first item is equivalent to computing $\mathrm{P}(E|D=\mathcal{G}
_\sigma)$, which is $1$, according to the proof in
Proposition~\ref{SVA_PostProb_Trunc_Fq}. Analogously, the second item corresponds to
computing $\mathrm{P}(E^\prime|D=U)$, which in the same place is shown to be
greater or equal to $1-q\left(\frac{1}{2}\pm\frac{1}{2q}\right)^M$.
\qed
\end{proof}

Now for the not truncated case, we have the following
\begin{proposition}\label{SVA_Success_NotTrunc_Fq}
With the same notations and assumptions as in Proposition~\ref{SVA_PostProb_NotTrunc_Fq},
\begin{enumerate}
\item
If the samples are PLWE, Algorithm~\ref{SVAAlg_Fq} guesses correctly with
probability at least $p_0^M$.
\item
If the samples are uniform, Algorithm~\ref{SVAAlg_Fq} guesses correctly with
probability at least $1-q\left(\frac{1}{2}\pm\frac{1}{2q}\right)^M$.
\end{enumerate}
\end{proposition}
\begin{proof}
The first item is equivalent to computing $\mathrm{P}(E|D=\mathcal{G}
_\sigma)$, which now is at least $p_0^M$, according to the proof in
Proposition~\ref{SVA_PostProb_NotTrunc_Fq}. The second item is not affected by the
truncation so the bound is the same as for the truncated case, namely,
greater or equal to $1-q\left(\frac{1}{2}\pm\frac{1}{2q}\right)^M$.
\qed
\end{proof}

\subsubsection{An improvement for a higher number of samples.}\label{HNS_SVA}

The same considerations over the number of samples present apply to this attack.
So, in the remainder of this section the same techniques will be applied.

Thus, it is first required to revisit an acceptable threshold for the number
of expected executions of the \emph{small values attack} sub-process
to output \textbf{PLWE}. In Proposition~\ref{SVA_Success_NotTrunc_Fq}, the success
probability of guessing the PLWE distribution of the samples is
characterized by:
\[
\mathrm{P}_{\mathcal{G}_{\sigma}}(M_0) := \mathrm{P}(E | D = \mathcal{G}_{\sigma}) \geq p_0^{M_0}.
\]
Therefore, given $M$ samples, the sub-process must executed $c := \lfloor M/M_0
\rfloor$ times, thus providing an expected number of outputs, $T$, different
from \textbf{NOT PLWE} given by
\[
T := c \cdot \mathrm{P}_{\mathcal{G}_{\sigma}}(M_0) \geq c \cdot p_0^{M_0}
\]
in case the samples were actually PLWE.

With these values, we define the \emph{Extended Small Values Attack} in
Algorithm~\ref{Extended_SVAAlg_Fq}.

\begin{figure}[ht]
\centering
\begin{tabular}[c]{ll}
\hline
\textbf{Input:} & A collection of samples $S=\{(a_i(x),b_i(x))\}
_{i=1}^M \in R_q^2$ \\
& A choice $M_0$ for the number of samples of the sub-process \\
\textbf{Output:} & A guess into the distribution of the samples, either \\
 & \textbf{PLWE} or \textbf{UNIFORM}\\
\hline
\end{tabular}

\medskip

\begin{itemize}
\item $T:=\lceil c \cdot p_0^{M_0} \rceil$
\item $C:=0$
\item \texttt{\emph{for}} $j\in \{0, ..., c - 1\}$ \texttt{\emph{do}}
	\begin{itemize}
	\item $result := SmallValuesAttack\left(S_{j} := \{(a_i(x),b_i(x))\}
                    _{i=j\cdot M_0 + 1}^{(j+1)\cdot M_0}\right)$
		\begin{itemize}
		\item \texttt{\emph{if}} result $\neq$ \textbf{NOT PLWE}
\texttt{\emph{then}}
			\begin{itemize}
			\item $C := C + 1$
			\end{itemize}
		\end{itemize}
	\end{itemize}
\item \texttt{\emph{if}} $C < T$ \texttt{\emph{then return}} \textbf{UNIFORM}
\item \texttt{\emph{else return}} \textbf{PLWE}
\end{itemize}
\hrule
\caption{Algorithm for Extended Small Values Attack}
\label{Extended_SVAAlg_Fq}
\end{figure}

Under this attack, we have the following success prediction probabilities,
which closely follow those in Subsection~\ref{ss.31}:
\begin{itemize}
\item
Case $D = \mathcal{G}_{\sigma}$: It is required to analyze the probability
$\mathrm{P}(C \geq T | D = \mathcal{G}_{\sigma})$. By definition, this can be seen as
$1 - \mathrm{P}(C < T | D = \mathcal{G}_{\sigma})$, which in turn can be interpreted in
terms of a cumulative binomial distribution. Therefore, we have that
\[
\mathrm{P}(C \geq T | D = \mathcal{G}_{\sigma}) := 1 - F(T-1, c,
\mathrm{P}_{\mathcal{G}_{\sigma}}
(M_0)) \geq 1 - F(T-1, c, p_0^{M_0})
\]
since the cumulative binomial distribution enjoys decreasing monotonicity on
the \emph{success probability} parameter.
\item
Case $D = U$: Now $\mathrm{P}(C < T | D = U)$ is yet again given in
terms of the cumulative binomial distribution, that is, the probability of a
binomial distribution of parameters $n = c$ and $p = \mathrm{P}_{U}(M_0)$, to
output any of the values $\{0,1,2,\dotsc,T-1\}$. Therefore, we have:
\[
\mathrm{P}(C < T | D = U) := F(T-1, c, \mathrm{P}_{U} (M_0)) \geq
F\left(T-1, c, 1 - \left(\frac{1}{2} \pm \frac{1}{2q}\right)^{M_0}\right),
\]
since the probability of success in case the samples were uniform (i.e.
$\mathrm{P}_{U} (M_0)$), in terms of the proof of Proposition~\ref{SVA_Success_NotTrunc_Fq},
verifies
\[
\mathrm{P}_{U} (M_0) = 1 - \mathrm{P}\left(\bigcup_{g=0}^{q-1} U_g\right) \geq
1 - \mathrm{P}\left(U_g\right) = 1 - \left(\frac{1}{2} \pm \frac{1}{2q}\right)^M,
\forall g \in \mathbb{F}_q
\]
and therefore can be upper-bounded by the probability that a particular $g_0 \in
\mathbb{F}_q$ verifies the condition for $M = M_0$ samples, which is $1 -
\left(\frac{1}{2} \pm \frac{1}{2q}\right)^{M_0}$.
\end{itemize}

With the above characterization of the success probabilities of each
distribution, we prove that Algorithm~\ref{Extended_SVAAlg_Fq} is successful:

\begin{proposition}\label{Extended_SVA_Fq}
Let $C, T, M_0$ be as specified in Algorithm~\ref{Extended_SVAAlg_Fq}. Then, if
$1 - \left(\frac{1}{2} \pm \frac{1}{2q}\right)^{M_0} < p_0^{M_0}$,
Algorithm~\ref{Extended_SVAAlg_Fq} is successful.
\end{proposition}
\begin{proof}
We need to prove that the algorithm verifies that:
\begin{multline*}
\frac{1}{2} \cdot (\mathrm{P}(C < T | D = U) +
\mathrm{P}(C \geq T | D = \mathcal{G}_{\sigma})) \geq  \\
\frac{1}{2} + \frac{1}{2}\cdot \left(F\left(T-1, c, 1 - \left(\frac{1}{2} \pm
\frac{1}{2q}\right)^{M_0}\right) - F(T-1, c, p_0^{M_0})\right) > \frac{1}{2}
\end{multline*}

Employing the decreasing monotonicity of the cumulative binomial distribution,
this is true as long as
\[
1 - \left(\frac{1}{2} \pm \frac{1}{2q}\right)^{M_0} < p_0^{M_0}.
\]
\qed
\end{proof}

This result provides a characterization of how $M_0$ should be chosen for the
attack to be applicable, although further refinements could be employed.

As in the case of the \emph{small set attack},
it is again noticeable that this extended algorithm does not make direct use
of features of $\mathbb{F}_q$, other than those related to the algorithm being
employed as sub-process. Therefore, this technique will also be applicable to
the extension to finite field extensions of Section~\ref{s.5} (with only having
to revisit, if necessary, the success probabilities of the sub-processes employed).

\subsection{Unbounded Small Values Attack, Revisited}\label{ss3.3}
We now turn our attention to the case where $2\overline{\sigma} \geq
\frac{q}{4}$. The value computed by evaluating at $\alpha$ a tentative error
polynomial, $b_i(\alpha) - a_i(\alpha)g$, for a sample $i$ and a guess
$g$ for $s(\alpha)$, no longer provides an automatic distinguisher, since this
value will lie outside the interval $[-\frac{q}{4}, \frac{q}{4})$ with
significant probability.

However, \cite{ELOS:2016:RCN} shows an attack to cover the previous
scenario, though considering a truncated Gaussian for the error polynomials.
This section provides a detailed overview of the implications of lifting the
latter requirement.

To construct this attack, the probability of the event $E_i$, defined as the
event in which $b_i(\alpha) - a_i(\alpha)g \in [-\frac{q}{4},
\frac{q}{4})$, for the $i$-th sample, and a guess $g$ for $s(\alpha)$,
needs to be re-calculated, and we will do it by means of a plain
``favorable vs total-cases'' analysis.

Since we work modulo $q$, we have (discrete) values inside the interval
$[0, q)$, which can equivalently be seen inside the interval $[-\frac{q}{2},
\frac{q}{2})$, by means of associating each value in the interval $[\frac{q}
{2}, q)$ to its unique congruent value in $[-\frac{q}{2}, 0)$.

With this transformation, its easy to see that for any $b = a \bmod{q}$, $b$
lies inside $[-\frac{q}{4}, \frac{q}{4})$ when
\[
a \in \left(\bigcup_{j \in \mathbb{N}} [-\frac{5\cdot q}{4} - j\cdot q, -
\frac{3\cdot q}{4} - j\cdot q)\right) \cup [-\frac{q}{4}, \frac{q}{4})
\cup \left(\bigcup_{j \in \mathbb{N}} [\frac{3\cdot q}{4} + j\cdot q, \frac{5\cdot q}
{4} + j\cdot q)\right).
\]
Hence, we can characterize the favorable cases by the area of the image
distribution across the favorable space, namely,
\[
\int_{-\frac{q}
{4}}^{\frac{q}{4}} \mathcal{G}_{\overline{\sigma}} + \sum_{j \in \mathbb{N}} \int_{-
\frac{5\cdot q}{4} - j\cdot q}^{-\frac{3\cdot q}{4} - j\cdot q}
\mathcal{G}_{\overline{\sigma}} + \sum_{j \in \mathbb{N}} \int_{\frac{3\cdot q}{4} +
j\cdot q}^{\frac{5\cdot q}{4} + j\cdot q} \mathcal{G}_{\overline{\sigma}}.
\]
Now, the area of this distribution is symmetric across $0$, so both integrals
can be truncated to its positive side, eventually yielding the following
probability:
\[
\mathrm{P}(E_i | D = \mathcal{G}_{\sigma}) =
\left(\int_{0}^{\infty}
\mathcal{G}_{\overline{\sigma}}\right)^{-1} \cdot \left( \int_{0}^{\frac{q}{4}}
\mathcal{G}_{\overline{\sigma}} + \sum_{j \in \mathbb{N}} \int_{\frac{3\cdot q}{4} +
j\cdot q}^{\frac{5\cdot q}{4} + j\cdot q}
\mathcal{G}_{\overline{\sigma}}\right).
\]
Note that, the implication of the truncated behavior of the Gaussian
distribution appears on the above probability: in a truncated setting, the
integral sum has a finite span, which corresponds to the finite number of
intervals inside $[-2\overline{\sigma}, 2\overline{\sigma}]$.

Denoting the value of this probability as $\mathrm{P}(E_i | D =
\mathcal{G}_{\sigma}) = \frac{1}{2} + \delta$, and
$\ell$ as the number of samples provided, the following attack applies:

\begin{figure}[ht]
\centering
\begin{tabular}[c]{ll}
\hline
\textbf{Input:} & A collection of samples $S=\{(a_i(x),b_i(x))\}
_{i=1}^\ell \in R_q^2$ \\
 & A value $\delta$  \\
\textbf{Output:} & A guess into the distribution of the samples, either \\
 & \textbf{PLWE} or \textbf{UNIFORM}\\
\hline
\end{tabular}

\medskip

\begin{itemize}
\item $T:=\left\lceil\frac{1}{2}\left(\ell\cdot q + 2\ell \cdot \delta \pm \ell
          \cdot \left(1 - \frac{1}{q}\right)\right) \right\rceil$
\item $C:=0$
\item \texttt{\emph{for}} $g\in \mathbb{F}_q$ \texttt{\emph{do}}
	\begin{itemize}
	\item \texttt{\emph{for}} $(a_i(x),b_i(x))\in S$  \texttt{\emph{do}}
		\begin{itemize}
		\item \texttt{\emph{if}} $b_i(\alpha)-a_i(\alpha)g \in [-\frac{q}{4}, \frac{q}{4})$
\texttt{\emph{then}}
			\begin{itemize}
			\item $C := C + 1$
			\end{itemize}
		\end{itemize}
	\end{itemize}
\item \texttt{\emph{if}} $C < T$ \texttt{\emph{then return}} \textbf{UNIFORM}
\item \texttt{\emph{else return}} \textbf{PLWE}
\end{itemize}
\hrule
\caption{Algorithm for Unbounded Small Values Attack over $\mathbb{F}_q$}
\label{USVAAlg_Fq}
\end{figure}
In Algorithm~\ref{USVAAlg_Fq}, $T$ represents the number of times that, on
average, one would expect the event $E_i$ to hold true when the Algorithm
is executed, should the samples actually come from the PLWE distribution.

To understand the value for this threshold $T$, consider that all candidates
$g$ from $\mathbb{F}_q$, except one, will not be the desired value
$g^\ast=s(\alpha)$, and therefore the \emph{tentative} error value
\[
b(\alpha) - g\cdot a(\alpha) = (s(\alpha) - g)\cdot a(\alpha) + e(\alpha)
\]
becomes the sum of a discrete uniform term in $\mathbb{F}_q$ and a discrete
Gaussian element of mean $0$ and variance $\overline{\sigma}^2$ over
$\mathbb{F}_q$. Therefore applying Lemma~\ref{SumProb}, we have that
the probability for that sum to lie inside $[-\frac{q}{4}, \frac{q}{4})$,
for any given sample $i$ and guess $g\neq g^\ast$, is $\frac{1}{2} \pm
\frac{1}{2q}$.

Moreover, for one value, namely, the correct guess $g^\ast$, we have that the
probability of the samples to lie inside the objective interval is $\frac{1}
{2} + \delta$, for each sample $i$, by the hypothesis.

Therefore, the number of expected values inside the desired interval can be
calculated as
\[
T:=\left\lceil \ell(q-1)\left(\frac{1}{2}\pm\frac{1}{2q}\right) +
                \ell\left(\frac{1}{2}+\delta\right) \right\rceil =
    \left\lceil\frac{1}{2}\left(\ell\cdot q + 2\ell \cdot \delta \pm \ell
    \cdot \left(1 - \frac{1}{q}\right)\right) \right\rceil.
\]

Under this attack, we have the following success prediction probabilities:
\begin{itemize}
\item
Case $D = U$: The probability $\mathrm{P}(C < T | D = U)$ is
given by the cumulative binomial distribution, that is, the probability of a
binomial distribution of parameters $n = \ell\cdot q$ and $p = \frac{1}{2} \pm
\frac{1}{2q}$, to output any of the values $\{0,1,2,\dotsc,T-1\}$. Therefore, we
have:
\[
\mathrm{P}(C < T | D = U) = F\left(T-1, \ell\cdot q, \frac{1}{2} \pm
\frac{1}{2q}\right).
\]
\item
Case $D = \mathcal{G}_{\sigma}$: Now we need to analyze the probability
$\mathrm{P}(C \geq T | D = \mathcal{G}_{\sigma})$. Let $C_{g^\ast}$ be defined
as the number of samples inside the objective interval, given the correct guess
$g^\ast$. Then:
\begin{equation}\label{eq.3}
\mathrm{P}(C \geq T | D = \mathcal{G}_{\sigma}) = \sum_{i=0}^{\ell}
\mathrm{P}(C - C_{g^\ast} \geq T - i | D = \mathcal{G}_{\sigma})
\cdot \mathrm{P}(C_{g^\ast} = i).
\end{equation}
The first term of the sum is yet again the complementary probability of a
cumulative binomial on all the guesses $g\neq g^\ast$, namely, those not
the correct one. Consequently, they fall under the case of
Lemma~\ref{SumProb}, and so its probability of satisfying the
event is $\frac{1}{2} \pm \frac{1}{2q}$.

The second term is the probability that a binomial distribution, with $n = \ell$
and $p = \frac{1}{2} + \delta$ outputs the value $i$. Therefore,
Equation~\eqref{eq.3} becomes
\[
\mathrm{P}(C \geq T | D = \mathcal{G}_{\sigma}) =
\sum_{i=0}^{\ell} \left (1 - F\left(T-i-1, \ell\cdot (q-1),\frac{1}{2} \pm
\frac{1}{2q}\right) \right) \cdot \mathrm{P}\left(B\left(\ell, \frac{1}{2} +
\delta\right) = i\right).
\]
\end{itemize}
Now, we prove that for a big enough modulo $q$, Algorithm~\ref{USVAAlg_Fq} is
successful:
\begin{lemma}\label{USVA_Success_Fq}
Let $C, T, \delta$ be as specified in Algorithm~\ref{USVAAlg_Fq}. Then, if
$\pm\frac{1}{2q} < \delta$, then the above algorithm is successful.
\end{lemma}
\begin{proof}
Since by hypothesis, both distributions, namely, uniform and PLWE, are
equiprobable, we need to prove that the following holds:
\begin{multline*}
    \frac{1}{2} \cdot (\mathrm{P}(C < T | D = U) +
                       \mathrm{P}(C \geq T | D = \mathcal{G}_{\sigma})) = \\
    \frac{1}{2} + \frac{1}{2}\cdot (\mathrm{P}(C < T | D = U) -
                                    \mathrm{P}(C < T | D = \mathcal{G}_{\sigma}))
    > \frac{1}{2}.
\end{multline*}
To do so, we will exploit the fact that, if $\frac{1}{2} \pm \frac{1}{2q}$ is
close enough to $\frac{1}{2}$, we get the result.

We compute the second term in the subtraction in terms of the number $C_g$
for all $g\neq g^\ast$, reasoning as above:
\begin{multline*}
\mathrm{P}(C < T | D = \mathcal{G}_\sigma) =
\sum_{i=0}^{\ell q-\ell} \mathrm{P}(C-C_g < T - i | D = \mathcal{G}_{\sigma}) \cdot
\mathrm{P}(C_g = i) \\
= \sum_{i=0}^{\ell q-\ell} F\left(T -i -1, \ell, \frac{1}{2} +
                                                        \delta\right) \cdot
\mathrm{P}\left(B\left(\ell q-\ell, \frac{1}{2} \pm \frac{1}{2q}\right) =
                                                        i\right).
\end{multline*}
Now for the first term we proceed similarly, keeping in mind that we are in
the case of the uniform distribution:
\[
\mathrm{P}(C < T | D = U) =
\sum_{i=0}^{\ell q-\ell} F\left(T -i -1, \ell, \frac{1}{2} \pm
\frac{1}{2q}\right) \cdot \mathrm{P}\left(B\left(\ell q-\ell, \frac{1}{2}
\pm \frac{1}{2q}\right) = i\right).
\]

Exploiting the decreasing monotonicity of the cumulative binomial present in
the first term and the fact that the second term is the same in both
probabilities, we conclude that Algorithm~\ref{USVAAlg_Fq} is successful if
\[
\Delta := \frac{1}{2} + \delta - \left(\frac{1}{2} \pm \frac{1}{2q}\right) =
\delta - \left(\pm \frac{1}{2q}\right) > 0.
\]
\qed
\end{proof}

\subsubsection{Success probability depending on the `distribution ratio'.}
The probability of success is depending on the formula laid out above, which is
based on the following computations:
\begin{enumerate}
\item $\int_{0}^{\infty} \mathcal{G}_{\overline{\sigma}} = \frac{1}{2}$

\item $\int_{0}^{\frac{q}{4}} \mathcal{G}_{\overline{\sigma}} =
       \frac{1}{2}\erf\left(\frac{q}{4\sqrt{2} \cdot \overline{\sigma}}\right)$
\item $\int_{\frac{3\cdot q}{4} + j\cdot q}^{\frac{5\cdot q}{4} + j\cdot q}
       \mathcal{G}_{\overline{\sigma}} =
       \frac{1}{2} \left(\erf\left(q\frac{5/4 + j}{\sqrt{2} \cdot \overline{\sigma}}\right) -
       \erf\left(q\frac{3/4 + j}{\sqrt{2} \cdot \overline{\sigma}}\right) \right)$
\end{enumerate}
where $\erf(x) = \frac{\sqrt{2}}{\pi} \cdot \int_{0}^{x} e^{-t^{2}} \cdot dt$
is the `error function'.
A common value, $r = \frac{q}{\sqrt{2} \cdot \overline{\sigma}}$, appears.
This value, referred to from now on as the \emph{distribution ratio}, reigns
over the probability value and therefore its convergence will be analyzed
based on it.

Since our assumption is $2\overline{\sigma}\geq \frac{q}{4}$, the
\emph{distribution ratio} satisfies $r \in (0, 4\sqrt{2}]$. Therefore, plotting
the function
\[
f(r) = \erf(r/4) + \sum_{j \in \mathbb{N}} \left(\erf(r \cdot
(5/4 + j)) - \erf(r \cdot (3/4 + j)) \right)
\]
and calculating when the condition $f(r) > \frac{1}{2} \pm \frac{1}{2q}$ is satisfied, for
values of $r$ in the interval $(0, 4\sqrt{2}]$, completely characterizes
the PLWE values that give rise to a successful attack.

\subsubsection{Practical instances.}
We provide a number of \textit{toy} examples to demonstrate the applicability
of the \emph{unbounded small values attack}:
\begin{enumerate}
\item
\begin{itemize}
    \item $N = 256$
    \item $q = 3329$.
    \item $\alpha = 3330$, order of $\alpha$ in $\mathbb{F}_q^\ast$ is $r = 2$.
    \item $\sigma = 8$.
    \item $\mathrm{P}(E_i | D = \mathcal{G}_{\sigma}) \approx 0.5983462$.
    \item $\Delta \approx 0.0981960$
\end{itemize}

\item
\begin{itemize}
    \item $N = 256$
    \item $q = 3677$.
    \item $\alpha = 3676$, order of $\alpha$ in $\mathbb{F}_q^\ast$ is $r = 2$.
    \item $\sigma = 8$.
    \item $\mathrm{P}(E_i | D = \mathcal{G}_{\sigma}) \approx 0.6377288$.
    \item $\Delta \approx 0.13759$
\end{itemize}

\item
\begin{itemize}
    \item $N = 256$
    \item $q = 2887$.
    \item $\alpha = 698$, order of $\alpha$ in $\mathbb{F}_q^\ast$ is $r = 3$.
    \item $\sigma = 8$.
    \item $\mathrm{P}(E_i | D = \mathcal{G}_{\sigma}) \approx 0.4999999$.
    \item $\Delta \approx 0.00017$
\end{itemize}

\item
\begin{itemize}
    \item $N = 256$
    \item $q = 4111$.
    \item $\alpha = 1055$, order of $\alpha$ in $\mathbb{F}_q^\ast$ is $r = 3$.
    \item $\sigma = 8$.
    \item $\mathrm{P}(E_i | D = \mathcal{G}_{\sigma}) \approx 0.5000000000009331$.
    \item $\Delta \approx 0.0001216$
\end{itemize}
\end{enumerate}
Regarding application of this attack to practical known instances deployed,
the situation is less likely to be fruitful.

Current PLWE schemes are mostly based upon cyclotomic rings, i.e.,
polynomial rings constructed from a cyclotomic polynomial. The specific
nature of these polynomials (i.e., the fact that their roots are primitive
roots of unity) causes that the order associated with their roots (the
\textit{conductor/index} of the cyclotomic polynomial) will be high enough
so that the attack cannot be computed, due to the size of the value of the
resulting standard deviation of the image Gaussian distribution.

Note that while this will remain the case over finite extensions of degree
$n$ of $\mathbb{F}_q$ (see Section~\ref{s.5}) as the elements considered
there in the cyclotomic case will also be $\frac{N}{gcd(n, N)}$-th primitive
roots of unity, the
use of higher-degree extension can make the attacks more applicable, as
both the order of the elements and the powers involved decrease.

\section{Small Set Attack for roots over finite field extensions of
\texorpdfstring{$\mathbb{F}_q$}{Lg}}\label{s.4}
In \cite{BBDN:2023:CPB} and \cite{BDNB:2023:TBC}, an attack was presented when
there existed a root $\alpha$ of $f(x)$ belonging to a proper $\mathbb{F}_q$
quadratic extension. The idea of that attack was to replace the set of the
likely values for $e(\alpha)$ by the set of their $\mathbb{F}_q$-traces, hence
drastically reducing the size of the small set at the cost of producing just
a decisional attack, instead of a search one.

However, the attack just described suggests a natural extension to settings
where $f(x)$ has a higher-than-quadratic-degree factor. This is actually what
we will do in the coming Subsections for the \emph{small set attack}, and in
Section~\ref{s.5} for the \emph{small values attack}.

Next we revisit traces and also introduce some tools, which will prove
important for the present and the coming Sections.

\subsection{Preliminaries}\label{ss.41}
Recall that $R_q$ is the ring $\mathbb{F} _q[x]/(f(x))$, with $f(x)$ a monic
irreducible polynomial over $\mathbb{Z}[x]$ of degree $N$. To introduce this
new setting, suppose now that the following conditions are met:
\begin{itemize}
\item
There exists $a\in\mathbb{F}_q$ such that the polynomial $x^n - a$ divides
$f(x)$ in $\mathbb{F}_q[x]$ and, therefore, $n < N$.
\item
The polynomial $x^n - a$ is irreducible over $\mathbb{F}_q[x]$, i.e., $f(x)$
has an irreducible factor of degree $n$.
\item
The element $a$ has a small multiplicative order in $\mathbb{F}_q^*$.
\end{itemize}
We choose $\alpha\in\mathbb{F}_{q^n}\setminus\mathbb{F}_q$ such that its minimal
polynomial over $\mathbb{F}_q$ is precisely $x^n - a$.
To pursue the adaptation of the attack described in
\cite{BBDN:2023:CPB} to the higher-degree setting we define now
$$
R_{q,0}=\{p(x)\in R_q: p(\alpha)\in\mathbb{F}_q\}.
$$
Then we have the following
\begin{proposition}\label{dim0}
The set $R_{q,0}$ is a ring and, as an $\mathbb{F}_q$-vector subspace of
$R_q$, has dimension $N-n+1$.
\end{proposition}
\begin{proof}
The fact that it is a subring of $R_q$ is obvious. As for the dimension we
observe that, given $p(x)=\sum_{i=0}^{N-1}p_ix^i$, we have that $p(x)\in
R_{q,0}$ if and only if
\[
p(\alpha) =          \sum_{j=0}^{N^\prime-1}a^jp_{nj} +
               \alpha\sum_{j=0}^{N^\prime-1}a^jp_{nj+1} +
\cdots + \alpha^{n-1}\sum_{j=0}^{N^\prime-1}a^jp_{nj+(n-1)}\in\mathbb{F}_q,
\]
with $N^\prime=\frac{N}{n}$, assuming for simplicity that $n\mid N$. Hence
$p(\alpha)\in\mathbb{F}_q$ if and only if
\begin{eqnarray*}
\sum_{j=0}^{N^\prime-1} a^jp_{nj + 1} & = & 0, \\
\sum_{j=0}^{N^\prime-1} a^jp_{nj + 2} & = & 0, \\
\cdots\\
\sum_{j=0}^{N^\prime-1} a^jp_{nj + n-1} & = & 0, \\
\end{eqnarray*}
and thus the result follows.
\qed
\end{proof}

We introduce now traces in this setting and show some interesting properties.
Let us begin with the following
\begin{proposition}\label{cvgen}
Suppose that $x^n-a$ divides $f(x)$ over $\mathbb{F}_q[x]$ for some
$n<N$. Assume that $x^n-a$ is irreducible over $\mathbb{F}_q[x]$ and let
$\alpha\in\mathbb{F}_{q^n}\setminus \mathbb{F}_q$ be a root of $f(x)$. Then,
for each $j\geq 1 $ such that $j$ is not a multiple of $n$ we have that
$\Tr(\alpha^j) = 0$.
\end{proposition}
\begin{proof}
First, assume that $1\leq j<n$ and let $\alpha_1,\dotsc,\alpha_n$ be the
conjugates of $\alpha_1:=\alpha$ in $\mathbb{F}_{q^n}$. We have
\[
\Tr(\alpha^j)=\sum_{j=1}^n\alpha_i^j,
\]
which is a symmetric expression evaluated at the conjugated roots. Moreover,
the symmetric polynomial $\Tr_j(x_1,\dotsc,x_n):=\sum_{i=1}^nx_i^j$ is
homogeneous of degree $j$. Hence, we can express
\[
\Tr_j(x_1,\dotsc,x_n)=g(s_1,\dotsc,s_n),
\]
where $g(x_1,\dotsc,x_n)$ is some polynomial and $s_i(x_1,\dotsc,x_n)$ is the
$i$-th elementary symmetric polynomial in the indeterminates $x_1,\dotsc,x_n$
(see \cite[1.12]{ST:2002:ANT}). Since the $\Tr_j$ has no independent term we
have that $g(0,\dotsc,0)=0$ and moreover, since $j<n$, $g(x_1,\dotsc,x_n)$ does
not contain monomial terms divisible by $x_n$. Finally, again by Cardano-Vieta,
since all the coefficients of $x^n-a$ vanish except for the first and the
last, we have that $s_i(\alpha_1,\dotsc,\alpha_n)=0$ for $1\leq i<n$ so that
$\Tr(\alpha^j)=0$.

Now, for $j>n$ not a multiple of $n$ we can write $j=tn+k$ with $1\leq k<n$ and
hence
$$
\Tr(\alpha^j)=a^t\Tr(\alpha^k)=0.
$$
\qed
\end{proof}

Another important feature of the trace operator is that it respects the
distribution of the candidate values that will be considered under the attacks:
\begin{proposition}\label{TraceDistHigher}
The trace operator respects the distribution of the tentative error values
$b(\alpha) - g\cdot a(\alpha)$, for any fixed $g \in \mathbb{F}_{q^n}$.
\end{proposition}

\begin{proof}
It is required to see that the traces of the evaluated tentative error
polynomials follow the same distributions as the case $\alpha \in \mathbb{F}_q$.

Define $a(x)=\sum_{i=0}^{N-1}a_ix^i$, $b(x)=\sum_{i=0}^{N-1}b_ix^i$. If the
samples are uniform, i.e., $a_i$, $b_i$ are uniform in $\mathbb{F}_q$
then, given $g \in \mathbb{F}_{q^n}$,
\[
\Tr(b(\alpha) - g \cdot a(\alpha)) = \sum_{j = 0}^{N-1} b_j \cdot \Tr(\alpha^j)
- \Tr\left(\sum_{j = 0}^{N-1} a_j \cdot \alpha^j \cdot g\right) = \sum_{j = 0}
^{N-1} b_j \cdot \Tr(\alpha^j) - \sum_{j = 0}^{N-1} a_j \cdot \Tr(\alpha^j
\cdot g),
\]
which is the sum of two linear combination of independent uniform distributions
in $\mathbb{F}_q$ and coefficients in $\mathbb{F}_q$ and, hence, by
Lemma~\ref{CorSumProb}, is a uniform distribution in $\mathbb{F}_q$.

Define $g^\ast=s(\alpha)$. If the samples are PLWE, then for the case $g =
g^\ast$, we have:
\[
\Tr(b(\alpha) - g^\ast\cdot a(\alpha)) = \Tr(e(\alpha)) = \sum_{j=0}^{N-1}
e_j \cdot \Tr(\alpha^j),
\]
which is the weighted (by $\Tr(\alpha^j) \in \mathbb{F}_q$) sum of (at most)
$N$ Gaussian random variables $e_j \in \mathbb{F}_q$ of mean $0$ and variance
$\sigma^2$ and, therefore, is a Gaussian distribution of mean $0$ and a certain
variance $\overline{\sigma}^2$ in $\mathbb{F}_q$.

Conversely, for the case $g \neq g^\ast$,
\[
\Tr(b(\alpha) - g \cdot a(\alpha)) = \Tr(e(\alpha) + (s(\alpha) - g) \cdot
a(\alpha)) = \Tr(e(\alpha)) + \sum_{j = 0}^N a_j \cdot \Tr(\alpha^j \cdot
(s(\alpha) - g)),
\]
which is the sum of a Gaussian distribution of mean 0 and a variance
$\overline{\sigma}^2$ and a uniform distribution in $\mathbb{F}_q$, and
therefore the result follows.
\qed
\end{proof}

\subsection{Layout of the generalized attack}
We present now the generalized attack divided into two stages, each one of them
giving rise to a specific algorithm.

The first stage is carried out by Algorithm~\ref{SSAAlg_Trace}, whose input is
a set of samples $S$ coming from $R_{q,0}\times R_q$ and allowing one to
distinguish whether such set is drawn either from the uniform or from the PLWE
distribution. In particular, the outcome of the Algorithm will be
one out of three possible results: \textbf{PLWE}, which means
that the samples come from a PLWE distribution; \textbf{NOT PLWE}, which means
the samples come from the uniform distribution; or \textbf{NOT ENOUGH
SAMPLES}, which means that the number of samples in the set is not enough to
provide a definitive conclusion.

\begin{figure}[ht]
\centering
\begin{tabular}[c]{ll}
\hline
\textbf{Input:} & A set of $M$ samples $S=\{(a_i(x),b_i(x))\}_{i=1}^M\in
R_{q,0}\times R_q$ \\
 & A look-up table $\Sigma$ of values appearing in $\frac{1}{n}
\Tr(b(\alpha) - a(\alpha)\cdot s(\alpha))$ \\
 & with probability $\geq p_0^r$ \\
\textbf{Output:} & \textbf{PLWE},\\
                 & or \textbf{NOT PLWE},\\
                 & or \textbf{NOT ENOUGH SAMPLES}\\
\hline
\end{tabular}

\medskip

\begin{itemize}
\item $G:=\emptyset$
\item \texttt{\emph{for}} $g\in \mathbb{F}_q$ \texttt{\emph{do}}
	\begin{itemize}
	\item \texttt{\emph{for}} $(a_i(x),b_i(x))\in S$  \texttt{\emph{do}}
		\begin{itemize}
		\item \texttt{\emph{if}} $\frac{1}{n}\left(\Tr(b_i(\alpha))
                                -a_i(\alpha)g\right)\notin\Sigma$
             \texttt{\emph{then next}} $g$
		\end{itemize}
	\item $G:=G\cup\{g\}$
	\end{itemize}
\item \texttt{\emph{if}} $G=\emptyset$ \texttt{\emph{then return}} \textbf{NOT
PLWE}
\item \texttt{\emph{if}} $|G|=1$ \texttt{\emph{then return}} \textbf{PLWE}
\item \texttt{\emph{if}} $|G|>1$ \texttt{\emph{then return}} \textbf{NOT ENOUGH
SAMPLES}
\end{itemize}
\hrule
\caption{Attack based on the size of the set of possible error values over
higher-degree extensions on $R_{q,0}$}
\label{SSAAlg_Trace}
\end{figure}

Observe, then, that if $(a(x),b(x)=a(x)s(x)+e(x))\in R_{q,0}\times R_q$ is a
PLWE sample, evaluating at $\alpha$ we have that
$\Tr(b(\alpha)-a(\alpha)s(\alpha)) = \Tr(e(\alpha))$. Computing the latter
\begin{equation}\label{eq.4}
\Tr(e(\alpha)) = \sum_{i=0}^{N-1} e_i \cdot \Tr(\alpha^i) =
\sum_{j=0}^{N^\prime-1} e_{nj} \cdot \Tr(\alpha^{nj}) =
n \cdot \sum_{j=0}^{N^\prime-1} e_{nj} \cdot a^j,
\end{equation}
where $N^\prime = \frac{N}{n}$ and, for simplicity, $n\mid N$.
Assuming that $r$ is the multiplicative order of $a$ in $\mathbb{F}_q^*$,
let us define $N^{\prime\prime} = \frac{N^\prime}{r}$ and, again for
simplicity, suppose that $r\mid N^\prime$. Thus we eventually obtain
\begin{equation}\label{partida3}
\Tr(e(\alpha)) = n
\sum_{i=0}^{r-1} a^i\cdot\sum_{j=0}^{N^{\prime\prime}-1} e_{n(jr + i)}.
\end{equation}

Now, two remarks are in order.
\begin{enumerate}
\item
Since $a(x)\in R_{q,0}$, this means that $a(\alpha)\in\mathbb{F}_q$, by
definition. Then
\[
\Tr(b(\alpha) - a(\alpha)\cdot s(\alpha)) = \Tr(b(\alpha)) -
a(\alpha)\Tr(s(\alpha)).
\]
Remark that in the previous expression the only unknown is precisely
$\Tr(s(\alpha))\in\mathbb{F}_q$. This fact allows our algorithm to loop over
$\mathbb{F}_q$ rather than over the full ring $R_q$.
\item
For fixed values $n$ and $r$, let us define $e_i^\ast=\sum_{j=0}^
{N^{\prime\prime}-1} e_{n(jr+i)}$, $0\leq i\leq r-1$, so that
Equation~\eqref{partida3} now reads
\begin{equation}\label{eq.5}
\Tr(e(\alpha)) = n\sum_{i=0}^{r-1} a^i e^\ast_i.
\end{equation}
Remark that each sum $e_i^\ast$ of coefficients can be seen as a sample from
a centered
discrete Gaussian of standard deviation less than or equal to $\sqrt{N^{\prime
\prime}}\sigma$ and, henceforth, we can list those elements that occur with
probability beyond $p_0$, which are at most $4\sqrt{N^{\prime\prime}}\sigma+1$.
\end{enumerate}
With these at hand, we can construct a look-up table
$\Sigma$ for all the possible values in Equation~\eqref{partida3} that happen
with probability beyond $p_0^r$. Such look-up table is defined as follows:
\[
\Sigma:=\left\{\sum_{i=0}^{r-1}a^ix_i:\; |x_i|\leq 2\sqrt{N^{\prime\prime}}
\sigma, x_i\in\mathbb{Z}\right\}.
\]
Observe that
\begin{equation}
|\Sigma|\leq \left(4\sqrt{N^{\prime\prime}}\sigma+1\right)^r.
\label{smregioncard}
\end{equation}

The second stage of the attack consists essentially in sampling from the oracle
presented to the adversary until they obtain a suitable number of samples from
$R_{q,0}\times R_q$. We provide Algorithm~\ref{alg4} in order to achieve this
task.

\begin{figure}[ht]
\centering
\begin{tabular}[c]{ll}
\hline
\textbf{Input:} & A distribution $X$ over $R_q^2$ \\
\textbf{Output:} & A sample $(a(x),b(x))\in R_{q,0}\times R_q$ \\
\hline
\end{tabular}

\medskip

\begin{itemize}
\item $count:=0$
\item \texttt{\emph{do}}
	\begin{itemize}
	\item	$(a(x),b(x)) \overset{X}{\leftarrow} R_q^2$
	\item $count := count+1$
	\end{itemize}
\item \texttt{\emph{until}} $a(x)\in R_{q,0}$
\item[]
\item \texttt{\emph{return}} $(a(x),b(x)), count$
\end{itemize}
\hrule
\caption{Random variable $X_0$}
\label{alg4}
\end{figure}

Algorithm~\ref{alg4} returns a pair $(a(x),b(x))\in R_{q,0}\times R_q$ by
repeatedly sampling from the original distribution $X$ keeping also track of the
number of invocations to $X$. Identifying this algorithm with a random variable
$X_0$ with values on $R_{q,0}\times R_q$, it is proved in \cite{BBDN:2023:CPB}
that $X$ is uniform (resp. PLWE) if and only if $X_0$ is uniform (resp. PLWE).

Notice that the variable $count$ stores the expected number of times one has to
run $X$ before succeeding. Based on a similar statement found in
\cite{BBDN:2023:CPB}, we have now the next
\begin{proposition}\label{p.37}
Notations as before, the variable $count$ is distributed as a
geometric random variable of first kind with success probability $q^{-(n-1)}$.
In particular the expected number of times is just its mean $E[count]=q^{n-1}$
\end{proposition}
\begin{proof}
Since $\mathrm{dim}_{\mathbb{F}_q}\left(R_q\right)=N$ and $\mathrm{dim}_
{\mathbb{F}_q}\left(R_{q,0}\right)=N-n+1$, then the probability that an element
taken uniformly from $R_q^2$ actually belongs to $R_{q,0}\times R_q$ is
precisely $q^{N-n+1}/q^N = q^{1-n}$. Now the average of a
geometric random variable of the first kind of success probability $\theta$ is
$\frac{1-\theta}{\theta}$ which is our case equals $q^{n-1}-1$ (see for instance
\cite[Appendix A]{BP:2015:MAI}).
\qed
\end{proof}

Since, as recommended, $q$ should be of order $\mathcal{O}(N^2)$ for the PLWE
cryptosystem to be feasible, the expected number of calls to PLWE oracle to
grant that our attack succeeds is of order $\mathcal{O}(N^{2(n-1)})$.

Note that the setting presented allows, theoretically, to direct the laid-out
attacks to arbitrary-degree extensions without any further restrictions. Despite
this fact, it is clear that the attacks just described rely heavily on the
degree of the extension and become almost infeasible for $n$ higher than $4$.

\subsection{Analysis of the algorithm}
The following results establish the reliability of
Algorithm~\ref{SSAAlg_Trace}, over any arbitrary degree extension. We will visit
in turn the cases where the Gaussian distributions are either truncated to width
$2\sigma$ or not truncated.

Let us begin with the case where Gaussian distributions are truncated. Then, we
have the following
\begin{proposition}\label{SSA_PostProb_Trunc_Trace}
Let $\Sigma$, $M$ and $r$ be defined as above, and assume $|\Sigma|<q$. Then:
\begin{enumerate}
\item
If Algorithm~\ref{SSAAlg_Trace} returns \textbf{PLWE} or \textbf{NOT ENOUGH SAMPLES},
then the samples are \textbf{PLWE} with probability at least $1-q\left(\frac
{|\Sigma|}{q}\right)^M$.
\item
If Algorithm~\ref{SSAAlg_Trace} returns \textbf{NOT PLWE}, then the samples are
uniform with probability $1$.
\end{enumerate}
\end{proposition}
\begin{proof}
Let us prove both items in turn.
\begin{enumerate}
\item
Denote with $E$ the event ``Algorithm~\ref{SSAAlg_Trace} returns \textbf{PLWE}
or \textbf{NOT ENOUGH SAMPLES}''. Applying Bayes' theorem,
\[
\mathrm{P}(D=\mathcal{G}_\sigma | E) = 1 - \mathrm{P}(D=U|E) =
1 - \frac{\mathrm{P}(E|D=U)}{\mathrm{P}(E|D=U) +
\mathrm{P}(E|D=\mathcal{G}_\sigma)}
\geq 1 - \frac{\mathrm{P}(E|D=U)}{\mathrm{P}(E|D=\mathcal{G}_\sigma)}.
\]
Let $U_g$ be the event that $\frac{1}{n}\Tr(b_i(\alpha)-a_i (\alpha)g)\in\Sigma$, for
all $i=1,\dotsc,M$, i.e., for all samples in $S$. Then
\[
\mathrm{P}(E|D=\mathcal{G}_\sigma) = \mathrm{P}\left(\bigcup_{g=0}^{q-1}
U_g\right) \geq \mathrm{P}(U_{g^\ast}),
\]
where $g^\ast = s(\alpha)$. Observe that $b_i(\alpha)-a_i (\alpha)g^\ast =
e_i(\alpha)$. Since we are dealing with the truncated case, we have
\[
\mathrm{P}(U_{g^\ast}) = \prod_{i=1}^M \mathrm{P}(\frac{1}{n}\Tr(e_i(\alpha))
\in\Sigma) = 1,
\]
because, according to Equation~\ref{eq.5},
$$
\frac{1}{n}\Tr(e_i(\alpha)) = \sum_{j= 0}^{r-1} a^j \cdot e^*_j\;\;
\in \Sigma,\;\;\forall i
$$
and each $e^\ast_j$ is a (truncated) Gaussian random variable.

Moreover, $\mathrm{P}(E|D=U)$ has the same bounds as in the proof of
Proposition~\ref{SSA_PostProb_Trunc_Fq}, since $\frac{1}{n}\Tr(b_i(\alpha)-a_i
(\alpha)g)$ is uniformly random in $\mathbb{F}_q$ by Proposition~$\ref{TraceDistHigher}$.
Combining results,
\[
\mathrm{P}(D=\mathcal{G}_\sigma | E) \geq 1 - \frac{\mathrm{P}(E|D=U)}
{\mathrm{P}(E|D=\mathcal{G}_\sigma)} = 1 -
\mathrm{P}(E|D=U)
\geq 1 - q\left(\frac{|\Sigma|}{q}\right)^M.
\]
\item
Denote again with $E^\prime$ the event ``Algorithm~\ref{SSAAlg_Trace} returns \textbf{NOT
PLWE}''. Applying Bayes' theorem,
\[
\mathrm{P}(D=U|E^\prime) = \frac{\mathrm{P}(E^\prime|D=U)}{\mathrm{P}(E^\prime|D=U) +
\mathrm{P}(E^\prime|D=\mathcal{G}_\sigma)}.
\]
Remark as before that the event $E^\prime$ is complementary to the event $E$,
defined in the previous item. Hence
\[
\mathrm{P}(E^\prime|D=\mathcal{G}_\sigma) = 1 - \mathrm{P}(E|D=\mathcal{G}
_\sigma) = 0.
\]
Consequently,
\[
\mathrm{P}(D=U|E^\prime) = 1.
\]
\end{enumerate}
\qed
\end{proof}

We turn now our attention to the non-truncated case, by stating the following
\begin{proposition}\label{SSA_PostProb_NotTrunc_Trace}
Let $\Sigma$, $M$ and $r$ be defined as above, and assume $|\Sigma|<qp_0^r$. Then:
\begin{enumerate}
\item
If Algorithm~\ref{SSAAlg_Trace} returns \textbf{PLWE} or \textbf{NOT ENOUGH SAMPLES},
then the samples are \textbf{PLWE} with probability at least $1-q\left(\frac
{|\Sigma|}{qp_0^r}\right)^M$.
\item
If Algorithm~\ref{SSAAlg_Trace} returns \textbf{NOT PLWE}, then the samples are
uniform with probability at least $1/2$ for large enough $M$.
\end{enumerate}
\end{proposition}
\begin{proof}
Let us prove both items in turn.
\begin{enumerate}
\item
As before, denote with $E$ the event ``Algorithm~\ref{SSAAlg_Trace} returns
\textbf{PLWE} or \textbf{NOT ENOUGH SAMPLES}''. Applying again Bayes' theorem,
\[
\mathrm{P}(D=\mathcal{G}_\sigma | E)
\geq 1 - \frac{\mathrm{P}(E|D=U)}{\mathrm{P}(E|D=\mathcal{G}_\sigma)}.
\]
Let $U_g$ be the event that $\frac{1}{n}\Tr(b_i(\alpha)-a_i (\alpha)g)\in\Sigma$, for
all $i=1,\dotsc,M$, i.e., for all samples in $S$. Then
\[
\mathrm{P}(E|D=\mathcal{G}_\sigma) = \mathrm{P}\left(\bigcup_{g=0}^{q-1}
U_g\right) \geq \mathrm{P}(U_{g^\ast}),
\]
where $g^\ast = s(\alpha)$. Since we are now dealing with the not truncated case, we have
\[
\mathrm{P}(U_{g^\ast}) = \prod_{i=1}^M \mathrm{P}(e_i(\alpha)\in\Sigma) \geq
p_0^{Mr}.
\]
Moreover, $\mathrm{P}(E|D=U)$ is the same as in the truncated case, since the
uniform distribution is not affected by the truncation. Combining results,
\[
\mathrm{P}(D=\mathcal{G}_\sigma | E) \geq 1 - \frac{\mathrm{P}(E|D=U)}
{\mathrm{P}(E|D=\mathcal{G}_\sigma)} \geq 1 -
\frac{\mathrm{P}(E|D=U)}{p_0^{Mr}}
\geq 1 - \frac{q\left(\frac{|\Sigma|}{q}\right)^M}{p_0^{Mr}} = 1 -
q\left(\frac{|\Sigma|}{qp_0^r}\right)^M.
\]
\item
Denote again with $E^\prime$ the event ``Algorithm~\ref{SSAAlg_Trace} returns \textbf{NOT
PLWE}''. Applying Bayes' theorem,
\[
\mathrm{P}(D=U|E^\prime) = \frac{\mathrm{P}(E^\prime|D=U)}{\mathrm{P}(E^\prime|D=U) +
\mathrm{P}(E^\prime|D=\mathcal{G}_\sigma)}.
\]
Remark as before that the event $E^\prime$ is complementary to the event $E$,
defined in the previous item. Hence
\[
\mathrm{P}(E^\prime|D=\mathcal{G}_\sigma) = 1 - \mathrm{P}(E|D=\mathcal{G}
_\sigma) \leq 1 - p_0^{Mr}.
\]
Consequently,
\[
\mathrm{P}(D=U|E^\prime) \geq \frac{\mathrm{P}(E^\prime|D=U)}{\mathrm{P}
(E^\prime|D=U) + 1-p_0^{Mr}},
\]
which is greater than $1/2$ if and only if $\mathrm{P}
(E^\prime|D=U)\geq 1-p_0^{Mr}$. Since
\[
\mathrm{P}(E^\prime|D=U) \geq 1-q\left(\frac{|\Sigma|}{q}\right)^M,
\]
for the conclusion to hold it is enough to take $M$ such that $q^{\frac{1}{M}}\frac{|\Sigma|}{q}\leq
p_0^r$.
\end{enumerate}
\qed
\end{proof}

All told, whenever $\left(4\sqrt{N^{\prime\prime}}\sigma+1\right)^r<qp_0^r$,
we have a decisional attack against the $R_{q,0}\times R_q$ PLWE problem. This
condition already imposes a restriction on $n$, since $r\leq q^n-1$ so that for
the look-up table to be checked in feasible time at each iteration, we must
restrict to roots such that $r$ is \emph{small enough} in relation to $n$,
which is difficult to test unless $n$ has a suitable size. This is one of the
reasons why we focus in the case $n=2,3$ or at most $4$. Nevertheless, from a
theoretical point of view, the presented generalization does not impose
restrictions on $n$.

\subsubsection{Success probabilities.}
Observe that from the discussion above, we can derive the success probability
of Algorithm~\ref{SSAAlg_Trace}. For the truncated case, we have the following
\begin{proposition}\label{SSA_Success_Trunc_Trace}
With the same notations and assumptions as in Proposition~\ref{SSA_PostProb_Trunc_Trace},
\begin{enumerate}
\item
If the samples are PLWE, Algorithm~\ref{SSAAlg_Trace} guesses correctly with
probability 1
\item
If the samples are uniform, Algorithm~\ref{SSAAlg_Trace} guesses correctly with
probability at least $1-q\left(\frac{|\Sigma|}{q}\right)^M$.
\end{enumerate}
\end{proposition}
\begin{proof}
The first item is equivalent to computing $\mathrm{P}(E|D=\mathcal{G}
_\sigma)$, which is $1$, according to the proof in
Proposition~\ref{SSA_PostProb_Trunc_Trace}. The second item is not affected by the
truncation so the bound is the same as for the truncated case, namely,
greater or equal to $1-q\left(\frac{|\Sigma|}{q}\right)^M$.
\qed
\end{proof}

For the not truncated case, we have the following
\begin{proposition}\label{SSA_Success_NotTrunc_Trace}
With the same notations and assumptions as in Proposition~\ref{SSA_PostProb_NotTrunc_Trace},
\begin{enumerate}
\item
If the samples are PLWE, Algorithm~\ref{SSAAlg_Trace} guesses correctly with
probability at least $p_0^{Mr}$.
\item
If the samples are uniform, Algorithm~\ref{SSAAlg_Trace} guesses correctly with
probability at least $1-q\left(\frac{|\Sigma|}{q}\right)^M$.
\end{enumerate}
\end{proposition}
\begin{proof}
The first item is equivalent to computing $\mathrm{P}(E|D=\mathcal{G}
_\sigma)$, which is at least $p_0^{Mr}$, according to the proof in
Proposition~\ref{SSA_PostProb_NotTrunc_Trace}. The second item is not affected by the
truncation so the bound is the same as for the truncated case, namely,
greater or equal to $1-q\left(\frac{|\Sigma|}{q}\right)^M$.
\qed
\end{proof}

\subsubsection{An improvement for a higher number of samples.}
As discussed
in Section~\ref{HNS_SSA}, the extension algorithm presented, which employed
Algorithm~\ref{SSAAlg_Fq} as a sub-process, did not employ any feature specific
of $\alpha \in \mathbb{F}_q$. As such, the same discussion and results presented
there can be applied to finite extensions of $\mathbb{F}_q$, giving rise to the
following algorithm:
\begin{figure}[H]
\centering
\begin{tabular}[c]{ll}
\hline
\textbf{Input:} & A collection of samples $S=\{(a_i(x),b_i(x))\}
                                           _{i=1}^M \in R_{q, 0} \times R_q$ \\
& A choice $M_0$ for the number of samples of the sub-process \\
\textbf{Output:} & A guess for the distribution of the samples, either \\
                 & \textbf{PLWE} or \textbf{UNIFORM}\\
\hline
\end{tabular}

\medskip

\begin{itemize}
\item $T:=\lceil \left\lfloor M/M_0\right\rfloor \cdot p_0^{M_0r} \rceil$
\item $C:=0$
\item \texttt{\emph{for}} $j\in \{0, ..., \left\lfloor M/M_0\right\rfloor - 1\}$ \texttt{\emph{do}}
	\begin{itemize}
	\item $result := \mathrm{SmallSetAttackTraces}\left(S_{j} := \{(a_i(x),b_i(x))\}
                         _{i=j\cdot M_0 + 1}^{(j+1)\cdot M_0}\right)$
		\begin{itemize}
		\item \texttt{\emph{if}} $result$ $\neq$ \textbf{NOT PLWE}
            \texttt{\emph{then}}
			\begin{itemize}
			\item $C := C + 1$
			\end{itemize}
		\end{itemize}
	\end{itemize}
\item \texttt{\emph{if}} $C < T$ \texttt{\emph{then return}} \textbf{UNIFORM}
\item \texttt{\emph{else return}} \textbf{PLWE}
\end{itemize}
\hrule
\caption{Algorithm for Extended Small Set Attack over finite extensions of $\mathbb{F}_q$}
\label{Extended_SSAAlg_Trace}
\end{figure}
Moreover, since the success probabilities detailed in
Proposition~\ref{SSA_Success_NotTrunc_Trace} for Algorithm~\ref{SSAAlg_Trace}
are the same as the ones in Proposition~\ref{SSA_Success_NotTrunc_Fq} for
Algorithm~\ref{SSAAlg_Fq}, all the results detailed in Section~\ref{HNS_SSA}
regarding the success probabilities of the extended attack apply here as well.

In particular, this means that:
\begin{itemize}
\item
Case $D = \mathcal{G}_{\sigma}$: We
have that
\[
\mathrm{P}(C \geq T | D = \mathcal{G}_{\sigma}) := 1 - F(T-1, \left\lfloor M/M_0\right\rfloor,
\mathrm{P}_{\mathcal{G}_{\sigma}}(M_0)) \geq 1 - F(T-1, \left\lfloor M/M_0\right\rfloor, p_0^{M_0 r})
\]

\item
Case $D = U$: We have:
\[
\mathrm{P}(C < T | D = U) := F(T-1, \left\lfloor M/M_0\right\rfloor, \mathrm{P}_{U}
(M_0)) \geq F(T-1, \left\lfloor M/M_0\right\rfloor,
1 - \left(|\Sigma|/q\right)^{M_0})
\]
\end{itemize}

\begin{proposition}
Let $C, T, M_0$ be as specified in Algorithm~\ref{Extended_SSAAlg_Trace}. Then, if
$1 - \left(|\Sigma|/q\right)^{M_0} < p_0^{M_0 r}$, Algorithm~\ref{Extended_SSAAlg_Trace}
is successful.
\end{proposition}
\begin{proof}
See proof of Proposition~\ref{Extended_SSA_Success_Fq}.
\qed
\end{proof}

\subsubsection{Practical instances.}\label{sss.421}
We evaluated the trace attack generalization over two \textit{toy} instances to
exemplify the practical applications of the above presented attacks. We have:
\begin{enumerate}
\item
\begin{itemize}
    \item $f = x^{23} - 2018x^{20} + x^{13} - 2018x^{10} + 2017x^{3} + 1$.
    \item $q = 4099$.
    \item $a = 2018$, order of $a$ in $\mathbb{F}_q^\ast$ is $r = 6$.
    \item $g(x) = x^3 - 2018$ irreducible in $\mathbb{F}_q[x]$ and a divisor of
$f(x)$.
    \item $N^\prime = 7$, $N^{\prime\prime} = 1$.
    \item $\sigma = 0.7$.
    \item $\left\lfloor 4\sqrt{N^{\prime\prime}}
           \sigma+1\right\rfloor ^r \approx 3010 < qp_0^r$.
\end{itemize}
Using these values, if Algorithm~\ref{SSAAlg_Trace} returns \textbf{PLWE}, then
the samples are indeed \textbf{PLWE} with probability $0.861$ for $350$ samples,
and with probability $0.998$ for $500$ samples.
\item
\begin{itemize}
    \item $f(x)=x^{23}-2017x^{20}+x^{13}-2017x^{10}+2018x^3+1$.
    \item $q = 4099$.
    \item $a = 2017$, order of $a$ in $\mathbb{F}_q^\ast$ is $r = 3$.
    \item $g(x) = x^3 - 2017$ irreducible in $\mathbb{F}_q[x]$ and a divisor of
$f(x)$.
    \item $N^\prime = 7$, $N^{\prime\prime} = 2$.
    \item $\sigma = 5/2$.
    \item $\left(4\sqrt{N^{\prime\prime}}
           \sigma+1\right)^r \approx 3471 < qp_0^r$.
\end{itemize}
For this second case, if Algorithm~\ref{SSAAlg_Trace} returns \textbf{PLWE}, then
the samples are indeed \textbf{PLWE} with probability $0.629$ for $350$ samples,
and with probability $0.993$ for $500$ samples.
\end{enumerate}

Notice that these simulations require less than one half of the number of samples
than in \cite{CLS:2017:ASR} and, moreover, we do not require the roots to belong
to $\mathbb{F}_q$. On the other hand, we are far from the improvement given in
\cite{CIV:2016:EDR} but it is convenient to point out that their attack only
applies to two particular (although infinite) family of trinomials.

We conducted also experiments to check the ability to find elements in
$R_{q,0}$. The details can be found in Appendix~\ref{ap.1}.

\section{Small Values Attack for roots over finite extensions of
\texorpdfstring{$\mathbb{F}_q$}{Lg}}\label{s.5}
In the previous section, we introduced a way to generalize root-based attacks,
focused on the \emph{small set attack} developed in \cite{ELOS:2015:PWI}.
In this section, we translate the same ideas to the \emph{Attack based on the
size of error values}, and provide with a generalization of these attacks for
arbitrary degree extensions of $\mathbb{F}_q$, exploiting the trace setting as
in the previous sections. As such, the general framework employed in the
previous section will apply here as well.

\subsection{Preliminaries}
Within this setting, we apply Proposition~\ref{cvgen}. Therefore, given the polynomial
$x^n - a$, we have:
\begin{enumerate}
    \item If $i$ is not a multiple of $n$, $\Tr(\alpha^{i}) = 0$.
    \item If $i$ is of the form $n\cdot j$, then $\Tr(\alpha^i) = \Tr(a^j) = a^j
\cdot \Tr(1) = n \cdot a^j$.
\end{enumerate}

The cases to consider are as follows, given the independent term $a$ of the
polynomial that builds the $n$-degree extension:
\begin{enumerate}
    \item Case $a = \pm 1$.
    \item Case $a \neq \pm 1$ and has small order $r$ modulo $q$.
    \item Case $a \neq \pm 1$ and does not have small order modulo $q$.
\end{enumerate}
When migrating to the trace setting under this attack, we have to study whether
the traces of the error values, once evaluated in $\alpha$ and taken modulo
$q$, are small. For convenience, let us remember Equation~\eqref{eq.4} where
the trace of an error value is presented:
\begin{equation*}
\Tr(e(\alpha)) = \sum_{i=0}^{N-1} e_i \cdot \Tr(\alpha^i) =
\sum_{j=0}^{N^\prime-1} e_{nj} \cdot \Tr(\alpha^{nj}) =
n \cdot \sum_{j=0}^{N^\prime-1} e_{nj} \cdot a^j,
\end{equation*}
where $N^\prime = \frac{N}{n}$ and, for simplicity, $n\mid N$.  Now we deal
in turn with the cases above:

\subsubsection{Case \texorpdfstring{$a = \pm 1$.}{Lg4}}
In this case, the value $\Tr(e(\alpha)) = n\sum_{j=0}^{N^\prime-1}
e_{nj} \cdot a^j$ is ($n$ times) a sum of $N^\prime$ Gaussians of mean $0$ and variance
$\sigma^2$, so it is a Gaussian of mean $0$ and variance $N^\prime \cdot \sigma^2$.
Therefore $\mathcal{G}
_{\overline{\sigma}}$ is a Gaussian with $\overline{\sigma}^2 = N^\prime \cdot \sigma^2$.

\subsubsection{Case \texorpdfstring{$a \neq \pm 1$ and has small
order $r$ modulo $q$.}{Lg5}}
Here, the small order property of the element $a$ is used to group the error
value $\Tr(e(\alpha))$ into $r$ packs of $\frac{N^\prime}{r}$ values (assuming
for simplicity that $r\mid N^\prime$). Each pack forms an
independent Gaussian of mean $0$ and variance $\frac{N^\prime}{r} \cdot \sigma^2$, and
thus the overall distribution is a weighted sum of such Gaussians, which generates
a Gaussian of mean $0$ and variance $\overline{\sigma}^2 = \sum_{i=0}^{r-1}
\frac{N^\prime}{r} \cdot \sigma^2 \cdot a^{2\cdot i}$.

\subsubsection{Case \texorpdfstring{$a \neq \pm 1$ and does not have
small order modulo $q$.}{Lg6}}
In the most general case, we just have a weighted sum of $N^\prime$ Gaussians, which
is itself a Gaussian of mean $0$ and variance $\overline{\sigma}^2 = \sum_{i=0}
^{N^\prime-1} \sigma^2 \cdot a^{2\cdot i}$.

\subsection{Small Values Attack}
Now, regardless of the above cases, the argument goes as follows: we will work
with the values of the Gaussian image generated by the trace of the evaluation
of the error polynomials inside $[-2\overline{\sigma}, 2\overline{\sigma}]$,
which have a combined mass probability of $p_0$. If $2\overline{\sigma} < \frac{q}{4}$ is
satisfied, Algorithm~\ref{SVAAlg_Trace} can be executed.

The introduction of higher-degree extensions on this attack not only allows to
generalize the setting in which this attack may be applicable, increasing the
reach of it, but also makes it more likely that the necessary condition
$2\overline{\sigma}<\frac{q}{4}$ holds.

Actually, introducing
the trace creates a narrower $\overline{\sigma}$ value, due to the fact that
fewer error terms are considered (most trace values are $0$), and
the fact that the order of the value $a$ is a divisor of the order of the root
$\alpha$.

In turn, the above attack severely increases the computational complexity of
the attack, on two fronts: The main loop over the guesses $g$ is displaced from
$\mathbb{F}_q$ to $\mathbb{F}_{q^n}$, and the requirement of evaluating the
trace on each sample and guess value $b_i(\alpha) - g \cdot a_i(\alpha)$.

To
deal with it, we resort to the auxiliary ring $R_{q,0}$ that was introduced
in \cite{BBDN:2023:CPB} and we have already used in Subsection~\ref{ss.41}.
With this modification, the two-fold increase in
computational complexity is alleviated, as the loop can now be restored to
$\mathbb{F}_q$, and the trace function need not be evaluated for each sample
and guess, it can be pre-computed for every sample before running the
algorithm.

\begin{figure}[ht]
\centering
\begin{tabular}[c]{ll}
\hline
\textbf{Input:} & A collection of samples $S=\{(a_i(x),b_i(x))\}
_{i=1}^M\subseteq R_{q,0} \times R_q$ \\
\textbf{Output:} & A guess $g\in\mathbb{F}_q$ for $\Tr(s(\alpha))$,\\
 & or \textbf{NOT PLWE},\\
 & or \textbf{NOT ENOUGH SAMPLES}\\
\hline
\end{tabular}

\medskip

\begin{itemize}
\item $G:=\emptyset$
\item \texttt{\emph{for}} $g \in \mathbb{F}_q$ \texttt{\emph{do}}
    \begin{itemize}
        \item \texttt{\emph{for}} $(a_i(x),b_i(x))\in S$  \texttt{\emph{do}}
            \begin{itemize}
                  \item \texttt{\emph{if}} $\frac{1}{n}(\Tr(b_i(\alpha))-a_i
                                            (\alpha)g) \notin [-\frac{q}{4}, \frac{q}{4})$
                   \texttt{\emph{then next}} $g$
            \end{itemize}
        \item $G:=G\cup\{g\}$
    \end{itemize}
\item \texttt{\emph{if}} $G=\emptyset$ \texttt{\emph{then return}} \textbf{NOT
PLWE}
\item \texttt{\emph{if}} $G=\{g\}$ \texttt{\emph{then return}} $g$
\item \texttt{\emph{if}} $|G|>1$ \texttt{\emph{then return}} \textbf{NOT ENOUGH
SAMPLES}
\end{itemize}
\hrule
\caption{Attack based on the size of error values over higher-degree
extensions on $R_{q,0}$}
\label{SVAAlg_Trace}
\end{figure}

\subsection{Analysis of the algorithm}
As in previous sections, we begin with the analysis of the truncated case:
\begin{proposition}\label{SVA_PostProb_Trunc_Trace}
Assume $2\overline{\sigma}<\frac{q}{4}$. Let
$M$ be the number of input samples.
\begin{itemize}
    \item[1.] If Algorithm~\ref{SVAAlg_Trace}
returns \textbf{NOT PLWE}, then the samples come from the uniform distribution,
with probability at least $1 - q\left(\frac{1}{2}\pm\frac{1}{2q}\right)^M$.
    \item[2.] If it outputs anything other than \textbf{NOT PLWE}, then the
samples are valid PLWE samples with probability $1$.

\end{itemize}
\end{proposition}
\begin{proof}
Let us prove both items in turn.
\begin{enumerate}
\item
As before, denote with $E$ the event ``Algorithm~\ref{SVAAlg_Trace} returns
\textbf{PLWE} or \textbf{NOT ENOUGH SAMPLES}''. Applying again Bayes' theorem,
\[
\mathrm{P}(D=\mathcal{G}_\sigma | E) = 1 - \mathrm{P}(D=U|E) =
1 - \frac{\mathrm{P}(E|D=U)}{\mathrm{P}(E|D=U) +
\mathrm{P}(E|D=\mathcal{G}_\sigma)}
\geq 1 - \frac{\mathrm{P}(E|D=U)}{\mathrm{P}(E|D=\mathcal{G}_\sigma)}.
\]
Let $U_g$ be again the event that $\frac{1}{n}\Tr(b_i(\alpha)-a_i (\alpha)g)
\in\left[-\frac{q}{4},\frac{q}{4}\right)$, for
all $i=1,\dotsc,M$, i.e., for all samples in $S$. Then
\[
\mathrm{P}(E|D=\mathcal{G}_\sigma) = \mathrm{P}\left(\bigcup_{g=0}^{q-1}
U_g\right) \geq \mathrm{P}(U_{g^\ast}),
\]
where $g^\ast = s(\alpha)$. Observe that $b_i(\alpha)-a_i (\alpha)g^\ast =
e_i(\alpha)$. Since we are dealing with the truncated case and since
$\left[-2\overline{\sigma},2\overline{\sigma}\right)\subseteq\left[
-\frac{q}{4},\frac{q}{4}\right)$ we have
\[
\mathrm{P}(U_{g^\ast}) = \prod_{i=1}^M \mathrm{P}\left(\frac{1}{n}
\Tr(e_i(\alpha))\in\left[-\frac{q}{4},\frac{q}{4}\right)\right) = 1.
\]
as $\frac{1}{n}\Tr(e_i(\alpha))$ is a (truncated) Gaussian random variable of
mean 0 and variance $\overline{\sigma}^2$.

Moreover, $\mathrm{P}(E|D=U)$ has the same bounds as in the proof of
Proposition~\ref{SVA_PostProb_Trunc_Fq}, since $\frac{1}{n}\Tr(b_i(\alpha)-a_i
(\alpha)g)$ is uniformly random in $\mathbb{F}_q$ by Proposition $\ref{TraceDistHigher}$.
Combining results,
\[
\mathrm{P}(D=\mathcal{G}_\sigma | E) \geq 1 - \frac{\mathrm{P}(E|D=U)}
{\mathrm{P}(E|D=\mathcal{G}_\sigma)} = 1 - \mathrm{P}(E|D=U)
= 1 - q\left(\frac{1}{2}\pm\frac{1}{2q}\right)^M.
\]
\item
Denote again with $E^\prime$ the event ``Algorithm~\ref{SSAAlg_Fq} returns \textbf{NOT
PLWE}''. Applying Bayes' theorem,
\[
\mathrm{P}(D=U|E^\prime) = \frac{\mathrm{P}(E^\prime|D=U)}{\mathrm{P}(E^\prime|D=U) +
\mathrm{P}(E^\prime|D=\mathcal{G}_\sigma)}.
\]
Remark as before that the event $E^\prime$ is complementary to the event $E$,
defined in the previous item. Hence
\[
\mathrm{P}(E^\prime|D=\mathcal{G}_\sigma) = 1 - \mathrm{P}(E|D=\mathcal{G}
_\sigma) = 0.
\]
Consequently,
\[
\mathrm{P}(D=U|E^\prime) = 1.
\]
\end{enumerate}
\qed
\end{proof}
Meanwhile, for the non truncated case
\begin{proposition}\label{SVA_PostProb_NotTrunc_Trace}
Assume $2\overline{\sigma}<\frac{q}{4}$ and $\frac{1}{2}\pm\frac{1}{2q}<p_0$
(which is the case for $q>2$). Let $M$ be the number of input samples.
\begin{itemize}
    \item[1.] If Algorithm~\ref{SVAAlg_Trace}
returns \textbf{NOT PLWE}, then the samples come from the uniform distribution,
with probability at least $1 - \frac{q\left(\frac{1}{2}\pm\frac{1}{2q}\right)^M}{p_0^{M}}$.
    \item[2.] If it outputs anything other than \textbf{NOT PLWE}, then the
samples are valid PLWE samples with probability at least $1/2$ for large enough $M$.

\end{itemize}
\end{proposition}
\begin{proof}
Let us prove both items in turn.
\begin{enumerate}
\item
As before, denote with $E$ the event ``Algorithm~\ref{SVAAlg_Trace} returns
\textbf{PLWE} or \textbf{NOT ENOUGH SAMPLES}''. As before,
\[
\mathrm{P}(D=\mathcal{G}_\sigma | E)
\geq 1 - \frac{\mathrm{P}(E|D=U)}{\mathrm{P}(E|D=\mathcal{G}_\sigma)}.
\]
Let $U_g$ be again the event that $\frac{1}{n}\Tr(b_i(\alpha)-a_i (\alpha)g)
\in\left[-\frac{q}{4},\frac{q}{4}\right)$, for
all $i=1,\dotsc,M$, i.e., for all samples in $S$. Then
\[
\mathrm{P}(E|D=\mathcal{G}_\sigma) = \mathrm{P}\left(\bigcup_{g=0}^{q-1}
U_g\right) \geq \mathrm{P}(U_{g^\ast}),
\]
where $g^\ast = s(\alpha)$. Since we are now dealing with the not truncated
case and since $\left[-2\overline{\sigma},2\overline{\sigma}\right)\subseteq
\left[-\frac{q}{4},\frac{q}{4}\right)$ we have
\[
\mathrm{P}(U_{g^\ast}) = \prod_{i=1}^M \mathrm{P}\left(e_i(\alpha)\in\left[
-\frac{q}{4},\frac{q}{4}\right)\right) \geq p_0^{M}.
\]
Moreover, $\mathrm{P}(E|D=U)$ is the same as in the truncated case, since the
uniform distribution is not affected by the truncation. Combining results,
\[
\mathrm{P}(D=\mathcal{G}_\sigma | E) \geq 1 - \frac{\mathrm{P}(E|D=U)}
{\mathrm{P}(E|D=\mathcal{G}_\sigma)} \geq 1 - \frac{\mathrm{P}(E|D=U)}{p_0^{M}}
\geq 1 - \frac{q\left(\frac{1}{2}\pm\frac{1}{2q}\right)^M}{p_0^{M}}.
\]
\item
Denote again with $E^\prime$ the event ``Algorithm~\ref{SVAAlg_Trace} returns
\textbf{NOT PLWE}''. As before,
\[
\mathrm{P}(D=U|E^\prime) = \frac{\mathrm{P}(E^\prime|D=U)}{\mathrm{P}(E^\prime|D=U) +
\mathrm{P}(E^\prime|D=\mathcal{G}_\sigma)}.
\]
Remark as before that the event $E^\prime$ is complementary to the event $E$,
defined in the previous item. Hence
\[
\mathrm{P}(E^\prime|D=\mathcal{G}_\sigma) = 1 - \mathrm{P}(E|D=\mathcal{G}
_\sigma) \leq 1 - p_0^{M}.
\]
Consequently,
\[
\mathrm{P}(D=U|E^\prime) \geq \frac{\mathrm{P}(E^\prime|D=U)}{\mathrm{P}
(E^\prime|D=U) + 1-p_0^{M}},
\]
which is greater than $1/2$ if and only if $\mathrm{P}
(E^\prime|D=U)\geq 1-p_0^{M}$. Since
\[
\mathrm{P}(E^\prime|D=U) \geq 1-q\left(\frac{1}{2}\pm\frac{1}{2q}\right)^M,
\]
for the conclusion to hold it is enough to take $M$ such that $q^{\frac{1}{M}}
\left(\frac{1}{2}\pm\frac{1}{2q}\right)\leq p_0$.
\end{enumerate}
\qed
\end{proof}

\subsubsection{Success probabilities.}
Observe that from the discussion above, we can derive the success probability
of Algorithm~\ref{SVAAlg_Trace}. For the truncated case, we have the following
\begin{proposition}\label{SVA_Success_Trunc_Trace}
With the same notations and assumptions as in Proposition~\ref{SVA_PostProb_Trunc_Trace},
\begin{enumerate}
\item
If the samples are PLWE, Algorithm~\ref{SVAAlg_Trace} guesses correctly with
probability $1$.
\item
If the samples are uniform, Algorithm~\ref{SVAAlg_Trace} guesses correctly with
probability at least $1-q(\frac{1}{2}\pm\frac{1}{2q})^M$.
\end{enumerate}
\end{proposition}
\begin{proof}
The first item is equivalent to computing $\mathrm{P}(E|D=\mathcal{G}
_\sigma)$, which is equal to $1$, according to the proof in
Proposition~\ref{SVA_PostProb_Trunc_Trace}. The second item is not affected by the
truncation so the bound is the same as for the truncated case, namely,
greater or equal to $1-q(\frac{1}{2}\pm\frac{1}{2q})^M$.
\qed
\end{proof}

For the not truncated case, we have the following
\begin{proposition}\label{SVA_Success_NotTrunc_Trace}
With the same notations and assumptions as in Proposition~\ref{SVA_PostProb_NotTrunc_Trace},
\begin{enumerate}
\item
If the samples are PLWE, Algorithm~\ref{SVAAlg_Trace} guesses correctly with
probability at least $p_0^{M}$.
\item
If the samples are uniform, Algorithm~\ref{SVAAlg_Trace} guesses correctly with
probability at least $1-q(\frac{1}{2}\pm\frac{1}{2q})^M$.
\end{enumerate}
\end{proposition}
\begin{proof}
The first item is equivalent to computing $\mathrm{P}(E|D=\mathcal{G}
_\sigma)$, which is at least $p_0^{M}$, according to the proof in
Proposition~\ref{SVA_PostProb_NotTrunc_Trace}. The second item is not affected by the
truncation so the bound is the same as for the truncated case, namely,
greater or equal to $1-q(\frac{1}{2}\pm\frac{1}{2q})^M$.
\qed
\end{proof}

\subsubsection{An improvement for a higher number of samples.} As discussed in
Section~\ref{HNS_SVA}, the extension algorithm presented which employed
Algorithm~\ref{SVAAlg_Fq} as a sub-process did not employ any feature specific
of $\alpha \in \mathbb{F}_q$. As such, the same discussion and results
presented there can be applied to finite extensions of $\mathbb{F}_q$, giving
rise to the following algorithm:
\begin{figure}[H]
\centering
\begin{tabular}[c]{ll}
\hline
\textbf{Input:} & A collection of $M$ samples $S=\{(a_i(x),b_i(x))\}
                                           _{i=1}^M \in R_{q, 0} \times R_q$ \\
& A choice $M_0$ of the number of samples for the sub-process \\
\textbf{Output:} & A guess for the distribution of the samples, either \\
                 & \textbf{PLWE} or \textbf{UNIFORM}\\
\hline
\end{tabular}

\medskip

\begin{itemize}
\item $T:=\lceil \left\lfloor M/M_0\right\rfloor \cdot p_0^{M_0} \rceil$
\item $C:=0$
\item \texttt{\emph{for}} $j\in \{0, ..., \left\lfloor M/M_0\right\rfloor - 1\}$ \texttt{\emph{do}}
	\begin{itemize}
	\item $result := \mathrm{SmallValueAttackTraces}\left(S_{j} := \{(a_i(x),b_i(x))\}
                         _{i=j\cdot M_0 + 1}^{(j+1)\cdot M_0}\right)$
		\begin{itemize}
		\item \texttt{\emph{if}} $result$ $\neq$ \textbf{NOT PLWE}
            \texttt{\emph{then}}
			\begin{itemize}
			\item $C := C + 1$
			\end{itemize}
		\end{itemize}
	\end{itemize}
\item \texttt{\emph{if}} $C < T$ \texttt{\emph{then return}} \textbf{UNIFORM}
\item \texttt{\emph{else return}} \textbf{PLWE}
\end{itemize}
\hrule
\caption{Algorithm for Extended Small Values Attack over finite extensions of
$\mathbb{F}_q$}
\label{ExtendedSVAAlg_Trace}
\end{figure}
Moreover, since the success probabilities detailed in Proposition~\ref{SVA_Success_NotTrunc_Trace} for Algorithm~\ref{SVAAlg_Trace} are the same as the ones in Proposition~\ref{SVA_Success_NotTrunc_Fq} for Algorithm~\ref{SVAAlg_Fq}, all the results detailed in Section~\ref{HNS_SVA} regarding the success probabilities of the extended attack apply here as well.

In particular, this means that:
\begin{itemize}
\item
Case $D = \mathcal{G}_{\sigma}$: We
have that
\[
\mathrm{P}(C \geq T | D = \mathcal{G}_{\sigma}) := 1 - F(T-1, \left\lfloor M/M_0\right\rfloor,
\mathrm{P}_{\mathcal{G}_{\sigma}}(M_0)) \geq 1 - F(T-1, \left\lfloor M/M_0\right\rfloor, p_0^{M_0})
\]

\item
Case $D = U$: We have:
\[
\mathrm{P}(C < T | D = U) := F(T-1, \left\lfloor M/M_0\right\rfloor, \mathrm{P}_{U}
(M_0)) \geq F(T-1, \left\lfloor M/M_0\right\rfloor,
1-(\frac{1}{2}\pm\frac{1}{2q})^{M_0})
\]
\end{itemize}

\begin{proposition}
Let $C, T, M_0$ be as specified in Algorithm~\ref{ExtendedSVAAlg_Trace}. Then, if
$1-(\frac{1}{2}\pm\frac{1}{2q})^{M_0} < p_0^{M_0}$, Algorithm~\ref{ExtendedSVAAlg_Trace}
is successful.
\end{proposition}
\begin{proof}
See proof of Proposition~\ref{Extended_SVA_Fq}.
\qed
\end{proof}

\subsection{Unbounded Small Values Attack}
This scenario represents the case in which the associated bound $2 \cdot
\overline{\sigma} \leq \frac{q}{4}$ does not hold. Let $E_i$ denote the event
that $\Tr(b_i(\alpha) - g \cdot a_i(\alpha))\in [\frac{-q}{4}, \frac{q}{4})$
for the $i$-th sample and fixed $g$. Then we must compute the probability of
$E_i$ given a PLWE sample distribution. But this can be done in the same way
as in Section~\ref{ss3.3}, just considering the Gaussian distributions
derived from the application of the trace, as defined above.

Note that, as stated in Proposition~\ref{TraceDistHigher}, the trace operator
respects the distributions. Accordingly, the application of
the trace operator to the tentative error values in the cases the samples are
either PLWE or uniform does not modify the expected behavior of the
distribution they follow with respect to the original case of Section~\ref{s.3}.

To integrate the trace evaluation into the attack, it is required to inspect
the threshold value $T$, to verify that, under this new setting, it is still
a useful decisional limit. Now denoting $g^\ast = s(\alpha) \in \mathbb{F}
_{q^n}$, it is clear that for the any value $g\neq g^\ast$  we have
\begin{multline*}
\frac{1}{n} \cdot \Tr(b(\alpha) - g\cdot a(\alpha)) =
\frac{1}{n} \cdot \Tr((s(\alpha) - g)\cdot a(\alpha) + e(\alpha)) \\ =
\frac{1}{n} \cdot \Tr((s(\alpha) - g)\cdot a(\alpha)) + \frac{1}{n}\cdot
\Tr(e(\alpha)),
\end{multline*}
which, by Proposition~\ref{TraceDistHigher}, is the sum of
uniform value in $\mathbb{F}_{q}$, and a discrete Gaussian of mean $0$ and
variance $\overline{\sigma}^2$ in $\mathbb{F}_q$, respectively. Thus,
we arrive to the same conditions as in Section~\ref{ss3.3}.

Denoting the probability as $\mathrm{P}(E_i | D = \mathcal{G}_{\sigma}) :=
\frac{1}{2} + \delta$, and referring to $\ell$ as the number of samples
provided, the attack is described in Algorithm~\ref{USVAAlg_Trace}.

\begin{figure}[ht]
\centering
\begin{tabular}[c]{ll}
\hline
\textbf{Input:} & A collection of samples $S := \{(a_i(x),b_i(x))\}
_{i=1}^\ell \subseteq R_{q,0} \times R_q$ \\
 & according to a certain distribution. \\
 & A value $\delta$ representing the difference over $\frac{1}{2}$. \\
\textbf{Output:} & A guess into the distribution of the samples, either \\
 & \textbf{PLWE} or \textbf{UNIFORM}.\\
\hline
\end{tabular}

\medskip

\begin{itemize}
\item $T:=\left\lceil\frac{1}{2}\left(\ell\cdot q + 2\ell \cdot \delta \pm \ell
          \cdot \left(1 - \frac{1}{q}\right)\right) \right\rceil$
\item $C:=0$
\item \texttt{\emph{for}} $g\in \mathbb{F}_q$ \texttt{\emph{do}}
	\begin{itemize}
	\item \texttt{\emph{for}} $(a_i(x),b_i(x))\in S$  \texttt{\emph{do}}
		\begin{itemize}
		\item \texttt{\emph{if}} $\frac{1}{n}(\Tr(b_i(\alpha))-a_i(\alpha)g) \in
[-\frac{q}{4}, \frac{q}{4})$
\texttt{\emph{then}}
			\begin{itemize}
			\item $C = C + 1$
			\end{itemize}
		\end{itemize}
	\end{itemize}
\item \texttt{\emph{if}} $C < T$ \texttt{\emph{then return}} \textbf{UNIFORM}
\item \texttt{\emph{else return}} \textbf{PLWE}
\end{itemize}
\hrule
\caption{PLWE Unbounded Small Values Attack}
\label{USVAAlg_Trace}
\end{figure}

The probability of success remain the same as in Section~\ref{ss3.3}. The only
differences come from the fact that the Gaussian image generated will be
smaller if we go to higher-degree extensions, though at the price of an
increase in the associated computational cost.

\begin{credits}
\subsubsection{\ackname} I. Blanco-Chac\'on is partially supported by the
Spanish National Research Plan, grant no.PID2022-136944NB-I00, by grant
PID2019-104855RB-I00, funded by MCIN/AEI/10.13039/501100011033 and by the
Universidad de Alcal\'a grant CCG20/IA-057. R. Dur\'an-D\'iaz is partially
supported by the Spanish National Research Plan, grant SATPQC
(PID2024-156099OB-I00), funded by MCIN / AEI / 10.13039 / 501100011033. R.
Mart\'{\i}n S\'anchez-Ledesma is partially supported by the PQReact
Project. This project has received funding from the European Union's
Horizon Europe research and innovation program under grant agreement
no. 101119547.

\subsubsection{\discintname}
The authors have no competing interests to declare that are
relevant to the content of this article.
\end{credits}

\bibliographystyle{splncs04}
\bibliography{plwe}

\appendix
\section{Experimental results}\label{ap.1}
We conducted also experiments to check the ability to find elements in
$R_{q,0}$, fixing the following parameters:
\begin{itemize}
\item We choose $N=64$ and the following irreducible polynomial in
$\mathbb{Z}[x]$:
\begin{multline*}
f(x)=
x^{64} + 1384\,x^{63} + 3830\,x^{62} + 1279\,x^{61} + 3341\,x^{60
} + 4034\,x^{59} + 3570\,x^{58} \\ + 1973\,x^{57}
\mbox{} + 4111\,x^{56} + 337\,x^{55} + 3980\,x^{54} + 789\,x^{53}
 + 2903\,x^{52} + 2876\,x^{51} \\ + 644\,x^{50}
\mbox{} + 4101\,x^{49} + 2944\,x^{48} + 959\,x^{47} + 1378\,x^{46
} + 2161\,x^{45} + 3710\,x^{44} \\ + 3040\,x^{43}
\mbox{} + 337\,x^{42} + 1329\,x^{41} + 3756\,x^{40} + 866\,x^{39}
 + 1790\,x^{38} + 3020\,x^{37} \\ + 2584\,x^{36}
\mbox{} + 266\,x^{35} + 140\,x^{34} + 1952\,x^{33} + 1920\,x^{32}
 + 3463\,x^{31} + 2540\,x^{30} \\ + 1686\,x^{29}
\mbox{} + 1825\,x^{28} + 1988\,x^{27} + 126\,x^{26} + 3366\,x^{25
} + 2982\,x^{24} + 1958\,x^{23} \\ + 57\,x^{22}
\mbox{} + 3543\,x^{21} + 1773\,x^{20} + 711\,x^{19} + 3607\,x^{18
} + 2686\,x^{17} + 519\,x^{16} \\ + 1996\,x^{15}
\mbox{} + 4106\,x^{14} + 3785\,x^{13} + 2775\,x^{12} + 1884\,x^{
11} + 2988\,x^{10} + 1605\,x^{9} \\ + 1819\,x^{8}
\mbox{} + 3000\,x^{7} + 3084\,x^{6} + 30\,x^{5} + 927\,x^{4} +
3110\,x^{3} + 2423\,x^{2} + 248\,x + 468
\end{multline*}
\item $n=4$, $q=4133$, $a=733$, so that the polynomial $g(x)=x^4-733$, is
irreducible in $\mathbb{F}_q$, but $g(x)|f(x)$ in $\mathbb{F}_q[x]$. The order
of $a$ in $\mathbb{F}_q^\ast$ is $r=4$.
\end{itemize}

Next we provide four elements in $R_{q,0}$, along with the number of
invocations to a uniformly random oracle, needed to generate each of them:
\begin{enumerate}
\item
Number of invocations: $25522319025$. The element is:
\begin{multline*}
524x^{63} + 3572x^{62} + 1163x^{61} + 3674x^{60} + 1684x^{59} + 3017x^{58} + 2292x^{57} + 2054x^{56} +  \\
 1126x^{55} + 143x^{54} + 2724x^{53} + 1763x^{52} + 3526x^{51} + 3950x^{50} + 75x^{49} + 1417x^{48} +  \\
 296x^{47} + 977x^{46} + 3383x^{45} + 2113x^{44} + 3292x^{43} + 1557x^{42} + 2369x^{41} + 231x^{40} +  \\
 609x^{39} + 2452x^{38} + 1369x^{37} + 2802x^{36} + 4032x^{35} + 2059x^{34} + 4021x^{33} + 2448x^{32} +  \\
 342x^{31} + 1166x^{30} + 539x^{29} + 3599x^{28} + 2907x^{27} + 2083x^{26} + 1879x^{25} + 2551x^{24} +  \\
 3151x^{23} + 2031x^{22} + 1529x^{21} + 304x^{20} + 555x^{19} + 2362x^{18} + 1182x^{17} + 1890x^{16} +  \\
 3717x^{15} + 2592x^{14} + 1361x^{13} + 4123x^{12} + 51x^{11} + 3019x^{10} + 4082x^{9} + 1079x^{8} +  \\
 1247x^{7} + 2353x^{6} + 2315x^{5} + 972x^{4} + 2298x^{3} + 3474x^{2} + 3553x + 1883
\end{multline*}

\item
Number of invocations: $79228885824$. The element is:
\begin{multline*}
945x^{63} + 3353x^{62} + 1216x^{61} + 979x^{60} + 802x^{59} + 691x^{58} + 438x^{57} + 3922x^{56} +  \\
 2532x^{55} + 3362x^{54} + 3378x^{53} + 1714x^{52} + 1625x^{51} + 2762x^{50} + 1708x^{49} + 2722x^{48} +  \\
 3828x^{47} + 3220x^{46} + 3342x^{45} + 1421x^{44} + 1599x^{43} + 112x^{42} + 3893x^{41} + 2618x^{40} +  \\
 3742x^{39} + 1342x^{38} + 37x^{37} + 995x^{36} + 260x^{35} + 2244x^{34} + 2638x^{33} + 141x^{32} +  \\
 1145x^{31} + 1104x^{30} + 3886x^{29} + 3115x^{28} + 856x^{27} + 1159x^{26} + 2385x^{25} + 1265x^{24} +  \\
 3609x^{23} + 3596x^{22} + 3689x^{21} + 2885x^{20} + 1665x^{19} + 411x^{18} + 2247x^{17} + 1700x^{16} +  \\
 3362x^{15} + 2802x^{14} + 2342x^{13} + 1449x^{12} + 132x^{11} + 1247x^{10} + 3072x^{9} + 2475x^{8} +  \\
 2811x^{7} + 2459x^{6} + 3916x^{5} + 991x^{4} + 1975x^{3} + 3335x^{2} + 1126x + 1083
\end{multline*}

\item
Number of invocations: $104055278839$. The element is:
\begin{multline*}
4125x^{63} + 3199x^{62} + 1098x^{61} + 162x^{60} + 1695x^{59} + 1736x^{58} + 2288x^{57} + 3874x^{56} +  \\
 1794x^{55} + 1281x^{54} + 1165x^{53} + 3769x^{52} + 1938x^{51} + 2190x^{50} + 1958x^{49} + 3578x^{48} +  \\
 3234x^{47} + 382x^{46} + 201x^{45} + 3963x^{44} + 1463x^{43} + 1964x^{42} + 3641x^{41} + 3445x^{40} +  \\
 3234x^{39} + 1866x^{38} + 3279x^{37} + 3788x^{36} + 2116x^{35} + 2851x^{34} + 3882x^{33} + 295x^{32} +  \\
 1267x^{31} + 4009x^{30} + 1655x^{29} + 3892x^{28} + 1027x^{27} + 1389x^{26} + 1723x^{25} + 4019x^{24} +  \\
 2587x^{23} + 3922x^{22} + 3290x^{21} + 1641x^{20} + 2646x^{19} + 251x^{18} + 642x^{17} + 4075x^{16} +  \\
 27x^{15} + 50x^{14} + 1811x^{13} + 1647x^{12} + 64x^{11} + 3342x^{10} + 2699x^{9} + 2452x^{8} +  \\
 3198x^{7} + 2745x^{6} + 2232x^{5} + 1859x^{4} + 1341x^{3} + 802x^{2} + 2162x + 838
\end{multline*}

\item
Number of invocations: $144485523325$. The element is:
\begin{multline*}
1701x^{63} + 2817x^{62} + 3829x^{61} + 1416x^{60} + 1069x^{59} + 1079x^{58} + 1434x^{57} + 3596x^{56} +  \\
 1019x^{55} + 1490x^{54} + 1603x^{53} + 4046x^{52} + 3686x^{51} + 3328x^{50} + 3621x^{49} + 1976x^{48} +  \\
 3876x^{47} + 3670x^{46} + 1472x^{45} + 1442x^{44} + 1548x^{43} + 377x^{42} + 410x^{41} + 2707x^{40} +  \\
 964x^{39} + 499x^{38} + 3991x^{37} + 1122x^{36} + 2234x^{35} + 1313x^{34} + 271x^{33} + 3444x^{32} +  \\
 3496x^{31} + 2638x^{30} + 990x^{29} + 1222x^{28} + 488x^{27} + 2395x^{26} + 2825x^{25} + 2352x^{24} +  \\
 2150x^{23} + 3640x^{22} + 2095x^{21} + 2489x^{20} + 3771x^{19} + 2760x^{18} + 439x^{17} + 2518x^{16} +  \\
 2743x^{15} + 3259x^{14} + 3598x^{13} + 2738x^{12} + 1755x^{11} + 1917x^{10} + 3443x^{9} + 1570x^{8} +  \\
 3431x^{7} + 3689x^{6} + 3078x^{5} + 2266x^{4} + 3869x^{3} + 1526x^{2} + 822x + 2986
\end{multline*}
\end{enumerate}
Remark that, according to Proposition~\ref{p.37} and subsequent comments, the
expected number of invocations is $\mathcal{O}\left(N^{2(n-1)}\right)$, which
for this case amounts to $64^6\simeq 68719 \cdot 10^6$. The average value in
our experiment yields roughly $88323\cdot 10^6$ invocations, pretty near to
the foreseen value. However completing the computation took from around four
days (wall time) for the fastest sample to around three weeks (also wall time)
for the slowest.  Given these figures, it seems that an example with $n=5$ is
way beyond our possibilities.
\end{document}